\documentclass[journal,draftclsnofoot,onecolumn,12pt,twoside]{IEEEtran}
\usepackage{graphicx}
\usepackage{algorithm}
\usepackage{algorithmic}

\usepackage{booktabs}

\usepackage[colorlinks,
linkcolor=black,
anchorcolor=black,
citecolor=black,
urlcolor = blue
]{hyperref}

\usepackage{amssymb}
\usepackage{subfigure}
\usepackage{multicol}
\usepackage{multirow}
\usepackage{amsmath}
\usepackage{bm}
\usepackage{cite}
\usepackage{gensymb}
\usepackage{amsfonts}
\usepackage{mathrsfs}
\usepackage{amsmath}
\usepackage{color}
\usepackage{balance}
\usepackage{bm}
\usepackage{setspace}

\usepackage{stfloats}
\usepackage{float}
\usepackage{epstopdf}

\usepackage{array}
\newcommand{\PreserveBackslash}[1]{\let\temp=\\#1\let\\=\temp}
\newcolumntype{C}[1]{>{\PreserveBackslash\centering}p{#1}}
\usepackage[flushleft]{threeparttable}
\usepackage{stfloats}
\usepackage{mathrsfs}

\newtheorem{lemma}{Lemma}

\newtheorem{proposition}{\it Proposition}

\begin{document}

\bibliographystyle{IEEEtran}

\title{Intelligent Reflecting Surface for MIMO VLC: \\Joint Design of Surface Configuration and Transceiver Signal Processing}
\vspace{-1cm}

\author{Shiyuan Sun, Fang Yang, \emph{Senior Member, IEEE}, Jian Song, \emph{Fellow, IEEE}, and Rui Zhang, \emph{Fellow, IEEE}
	
\thanks{
	This work was supported by the National Natural Science Foundation of China (61871255), the Beijing National Research Center for Information Science and Technology (BNR2022RC01017), and the Fok Ying Tung Education Foundation.
%	\textit{(Corresponding author: Fang Yang)}
		
	Shiyuan Sun, Fang Yang, and Jian Song are with the Department of Electronic Engineering, Beijing National Research Center for Information Science and Technology, Tsinghua University, Beijing 100084, China, and also with the Key Laboratory of Digital TV System of Guangdong Province and Shenzhen City, Research Institute of Tsinghua University in Shenzhen, Shenzhen 518057, China 
	(e-mail: sunsy20@mails.tsinghua.edu.cn; fangyang@tsinghua.edu.cn; jsong@tsinghua.edu.cn).
	
	Rui Zhang is with the Department of Electrical and Computer
	Engineering, National University of Singapore, Singapore 117583 (e-mail:
	elezhang@nus.edu.sg).
	He is also with the School of Science and Engineering, the Chinese University of Hong Kong (Shenzhen), China, 518172 (e-mail: rzhang@cuhk.edu.cn).
}}

\maketitle
\begin{abstract}
With the capability of reconfiguring the wireless electromagnetic environment, intelligent reflecting surface (IRS) is a new paradigm for designing future wireless communication systems.
In this paper, we consider optical IRS for improving the performance of visible light communication (VLC) under a multiple-input and multiple-output (MIMO) setting.
Specifically, we focus on the downlink communication of an indoor MIMO VLC system and aim to minimize the mean square error (MSE) of demodulated signals at the receiver.
To this end, the MIMO channel gain of the IRS-aided VLC is first derived under the point source assumption, based on which the MSE minimization problem is then formulated subject to the emission power constraints.
Next, we propose an alternating optimization algorithm, which decomposes the original problem into three subproblems, to iteratively optimize the IRS configuration, the precoding and detection matrices for minimizing the MSE.
Moreover, theoretical analysis on the performance of the proposed algorithm in high and low signal-to-noise ratio (SNR) regimes is provided, revealing that the joint optimization process can be simplified in such special cases, and the algorithm's convergence property and computational complexity are also discussed.
Finally, numerical results show that IRS-aided schemes significantly reduce the MSE as compared to their counterparts without IRS, and the proposed algorithm outperforms other baseline schemes.
\end{abstract}


\IEEEpeerreviewmaketitle

%\vspace{-0.5cm}
\section{Introduction}
%\vspace{-0.1cm}
\IEEEPARstart{T}{o} cope with approximately 50 exabyte-per-month data traffic in the fifth-generation (5G) and beyond~\cite{traffic}, wireless communication technologies have been significantly advanced to offer higher spectral efficiency (SE) and energy efficiency (EE), more reliable signal transmission, lower transmission latency, etc.
Different from other traditional wireless communication technologies such as millimeter-wave (mmWave) and Terahertz (THz) communications~\cite{heath2016overview}, visible light communication (VLC) modulates signals in the untapped frequency band with wavelengths in the range of 480 nm$\sim$750 nm, and the non-coherent modulation scheme called intensity modulation (IM) / direct detection (DD) is typically employed in VLC.
Owing to the license-free merit of the operating band, abundant spectrum resources including about 400 THz can be safely exploited in VLC, which shows great potential for improving the wireless communication capacity~\cite{pathak2015visible,karunatilaka2015led}.
In the meanwhile, VLC is a green technology due to its low-power transceivers consisting of light-emitting diodes (LEDs) and photodetectors (PDs), and it far outperforms other radio frequency (RF) communications in terms of EE~\cite{obeed2019optimizing}.
Moreover, VLC has advantages in dense frequency reuse, ubiquitous LED devices, and inherent physical layer security, which therefore attracts great attention from academia and industry in developing new VLC techniques, e.g., VLC Consortium (VLCC), IEEE 802.15.7 group, and Home Gigabit Access (OMEGA)~\cite{burchardt2014vlc,chi2015visible,obeed2019optimizing}.

Despite the aforementioned benefits, VLC also faces critical challenges, especially in the case of multiple-input multiple-output (MIMO) VLC~\cite{obeed2019optimizing}.
Specifically, MIMO has been regarded as a revolutionary technique, which can offer multiplexing gain as well as diversity gain, and thus has been widely adopted in 3G to 5G wireless communication systems~\cite{heath2016overview}.
However, the mainstream IM/DD (de)modulation scheme in VLC imposes a real-valued and nonnegative constraint on signal amplitudes, which generally follows the typical Lambertian model that highly depends on the geometric locations of transceivers~\cite{obeed2019optimizing}.
Owing to that, considering a typical scenario where PDs are close to each other, the spatial signature of a specific LED on all PDs is almost homogenous, which results in poor multiplexing and diversity gains of the MIMO VLC~\cite{karunatilaka2015led}.
A more detailed explanation is provided in~\cite{ying2015joint}, where the singular value decomposition (SVD) of a $4\times 4$ MIMO VLC channel is carried out to show that its condition number is as large as 189.3, i.e., only 1 data stream can be supported rather than 4.
Generally, the performance of MIMO VLC can be improved by designing the shape of the receiver, e.g., the angle diversity receiver, the receiver with narrow field-of-view (FoV) PDs, the imaging receiver, etc~\cite{nuwanpriya2015indoor,chen2014high}.
On the other hand, this can also be accomplished by signal processing such as power imbalance and precoding~\cite{ying2015joint}.
These methods attempt to improve the MIMO VLC through adjustments of transceivers; however, far little work has been done from the view of optical channel modification.

Benefitting from the wireless environment configuration capability of the emerging intelligent reflecting surface (IRS)~\cite{wu2019towards}, the wireless channel evolves to be a controllable variable rather than an uncontrolled barrier that has to be overcome~\cite{RIS_overview}.
Basically, IRS is a two-dimensional planar surface composed of periodical artificial atoms, through which the impinging electromagnetic wave can produce controllable induction current patterns, and therefore wave steering, polarizing, and absorbing can be dynamically achieved~\cite{liaskos2018new,yang2016programmable}.
As a promising new paradigm, the research on the RF IRS has been in full swing over recent years, including its channel modeling and estimation~\cite{tang2021wireless,9133156}, IRS-aided MIMO communication~\cite{zhang2020capacity,RIS_DPC}, joint IRS configuration and precoding~\cite{wu2019intelligent}, IRS-enhanced physical layer security~\cite{RIS_Cooperation}, etc.

On the other hand, there has been intensive interest recently in optical IRS within the visible light frequency range, which can be physically implemented by mirror array-based design or meta-surface-based design~\cite{komar2017electrically, hail2019optical}.
Reflected channel characterizations of these two designs are elaborated in~\cite{abdelhady2020visible,9681888,najafi2019intelligent,2021intelligent}, and more comprehensive reviews of the optical IRS are provided in~\cite{9627820,Sun2021,9614037}.
Given the flexibility of the IRS, it has been widely utilized in various aspects to facililtate VLC, namely IRS deployed at the receiver to actively expand FoV~\cite{ndjiongue2021re, ndjiongue2021toward}, IRS deployed at the transmitter to accomplish beam steering~\cite{9614037}, and IRS deployed in the wireless environment to enhance SE/EE and alleviate blockages~\cite{sun_CL,9543660,sun2022joint}.
Moreover, it is shown in~\cite{Sun2021,9614037} that the optical IRS can potentially be used in the MIMO VLC system, where the MIMO channel can be modified by IRS and thereby the multiplexing gain is improved.

In this paper, we consider the downlink communication of an indoor MIMO VLC system, where the emission signal propagates to the receiver through the line-of-sight (LoS) path directly as well as the non-LoS (NLoS) paths specularly reflected by optical IRS.
Under the point source assumption, the perfect channel state information (CSI) is assumed known to the transmitter, and this can be achieved by various channel estimation techniques~\cite{chen2016adaptive,9257457}.
Then, a joint optimization of the IRS configuration and transceiver signal processing including MIMO precoding/detection is investigated for minimizing the receiver mean square error (MSE) of the IRS-aided MIMO VLC.
The main contributions of this paper are summarized as follows:
\begin{itemize}
	\item
	First, this paper presents a new IRS-aided MIMO VLC model.
	Since the frequency of the visible light is much higher than that of RF waves, the power density distribution of the reflected beam concentrates in a certain direction in VLC~\cite{9627820}, and therefore the configuration of each IRS unit can be specified by its association with different transceivers~\cite{sun2022joint}.
	In this extremely near-field scenario, the IRS-aided MIMO VLC channel gain is defined with respect to two association matrices, based on which the MSE minimization problem is then formulated.
	
	\item
	Next, after transforming the formulated problem to a more tractable form, an alternating optimization (AO) algorithm is proposed to solve the problem by iteratively solving three convex subproblems, namely the IRS configuration subproblem, the precoding subproblem, and the detection optimization subproblem.
	Moreover, theoretical analysis on the performance of the proposed AO algorithm in high and low signal-to-noise ratio (SNR) regimes is provided to draw essential insight and simplify the design, and the convergence and computational complexity of the proposed algorithm are also analyzed.
	
	\item
	Finally, the performance of the IRS-aided VLC and the proposed algorithm are numerically evaluated in an indoor environment.
	Several baselines are chosen for comparison, including the distance greedy and the random allocation scheme for IRS configuration, and the zero-forcing (ZF) and the minimum mean square error (MMSE) schemes for precoding.
	From simulation results, it is verified that the MSE and bit error rate (BER) can be significantly lowered by IRS, and the proposed algorithm outperforms other baselines.
\end{itemize}

%本文的结构简述如下
The rest of this paper is organized as follows. In Section~\ref{Sec:Channel}, the channel gains of LoS and NLoS paths are obtained based on the Lambertian model, with which the IRS-aided optical MIMO channel model is established. 
Then, the signal model including the modulation scheme and signal processing is specified in Section~\ref{Sec:Model}, where the MSE expression is derived and the optimization problem is formulated.
In Section~\ref{Sec:ProposedAlgorithm_minimizeMSE}, we propose an AO algorithm to minimize the MSE, and present the performance analysis of the proposed algorithm in two special cases with high and low SNR values, respectively.
Afterwards, numerical results are provided in Section~\ref{Sec:Num} to evaluate the performance of the proposed algorithm for IRS-aided VLC.
Finally, Section~\ref{Sec:Conclude} concludes this paper.

\textit{Notations:}
In this paper, scaler values, vectors, and matrices are denoted by normal letters such as $a$ (or $A$), boldface letters such as $\boldsymbol{a}$, boldface uppercase letters such as $\boldsymbol{A}$, respectively.
Particularly, $\textbf{I}_N$, $\textbf{0}$, and $\textbf{1}$ indicate the $N \times N$ identity matrix, an all-zero matrix, and an all-one matrix, respectively.
The special operators include the Kronecker product $\otimes$, Hadamard product $\odot$, matrix transpose $(\cdot)^T$, vectorization operator $\text{vec}(\cdot)$, trace operator $\text{tr}(\cdot)$, positive definite operator $\succ$, and the differential operator $\partial$ ($\nabla$).
$\mathbb{E} [\cdot]$ is used to represent the statistical expectation.
Besides, diag($[\textbf{a}_1, \cdots, \textbf{a}_n]$) is a block diagonal matrix, where the $i$-th main diagonal vector is given by $\textbf{a}_i,\ i = 1, 2, \cdots, n$.
%Besides, diag($[\textbf{a}_1, \cdots, \textbf{a}_n]$) is a block diagonal matrix, where the $i$-th block diagonal element is a column vector $\textbf{a}_i,\ i = 1, 2, \cdots, n$ while the off-diagonal elements are zeros.
The operators $\left\|\boldsymbol{A}\right\|_F$, $\left\|\boldsymbol{A}\right\|_r$, and $\left\|\boldsymbol{a}\right\|_1$ denote the Frobenius norm, the row norm, and the $l_1$ norm, respectively.
Moreover, the calligraphic letter $\mathcal{A}$ denotes a set with $|\mathcal{A}|$ representing the number of elements in the set; $\mathbb{R}$ and $\mathbb{R}_{+}$ denote the real number set and the real and positive number set, respectively.

%****************************************************信道******************************************************************
\section{Channel Model}
\label{Sec:Channel}
The LoS and NLoS channel components of the IRS-aided MIMO VLC system are described in this section, where the single antenna channel gain is given first based on the Lambertian model and the MIMO VLC channel matrix is then derived.
Particularly, the indices of the transmitter, receiver, and IRS unit are denoted by $n_t\in\mathcal{T} \triangleq \{1,\cdots, N_t\}$, $n_r\in\mathcal{R}\triangleq \{1,\cdots, N_r\}$, $n\in\mathcal{N}\triangleq \{1,\cdots, N\}$, with the cardinalities of $|\mathcal{T}| = N_t$, $|\mathcal{R}| = N_r$, and $|\mathcal{N}| = N$, respectively.

%\vspace{-0.3cm}
\subsection{Channel Gains of LoS/NLoS Paths}
\label{Subsec:Gain}
Before deriving the channel gain, key geometric parameters need to be specified since the channel gain depends on them in VLC.
Suppose that the locations of the $n_t$-th LED, the $n_r$-th PD, and the $n$-th IRS unit are $(n_{t,x}, n_{t,y}, n_{t,z})$, $(n_{r,x}, n_{r,y}, n_{r,z})$, and $(x_n, y_n, z_n)$, respectively in a three-dimensional (3D) Cartesian coordinate system.
The distance between each pair of LED and PD can be expressed as
\begin{equation}
	D_{n_r,n_t} = \sqrt{(n_{r,x} - n_{t,x})^2 + (n_{r,y}-n_{t,y})^2 + (n_{r,z}-n_{t,z})^2},
\end{equation}
and those of LED-to-IRS links and IRS-to-PD links are given by
\begin{align}
	d_{n,n_t} &= \sqrt{(x_n - n_{t,x})^2 + (y_n - n_{t,y})^2 + (z_n - n_{t,z})^2},\\
	d_{n_r,n} &= \sqrt{(n_{r,x} - x_n)^2 + (n_{r,y} - y_n)^2 + (n_{r,z} - z_n)^2}.
\end{align}

Without loss of generality, the normal vectors of transceivers are assumed to be perpendicular to the ground, and the angles of irradiance and incidence of each LoS path are given by
\begin{align}
	\Theta_{n_r,n_t} = \Phi_{n_r,n} = \arccos\left( \frac{n_{t,z}-n_{r,z}}{D_{n_r,n_t}} \right).
\end{align}
Similarly, the angle of irradiance of the NLoS path associated with the $n_t$-th LED and the $n$-th IRS unit can be expressed as
\begin{align}
	\theta_{n,n_t} &= \arccos\left( \frac{n_{t,z}-z_n}{d_{n,n_t}} \right),
\end{align}
and that at the $n_r$-th PD is given by
\begin{align}
	\phi_{n_r,n} = \arccos\left( \frac{z_n-n_{r,z}}{d_{n_r,n}} \right).
\end{align}

\begin{figure}[t]
	\centering
	\includegraphics[width=0.6\textwidth]{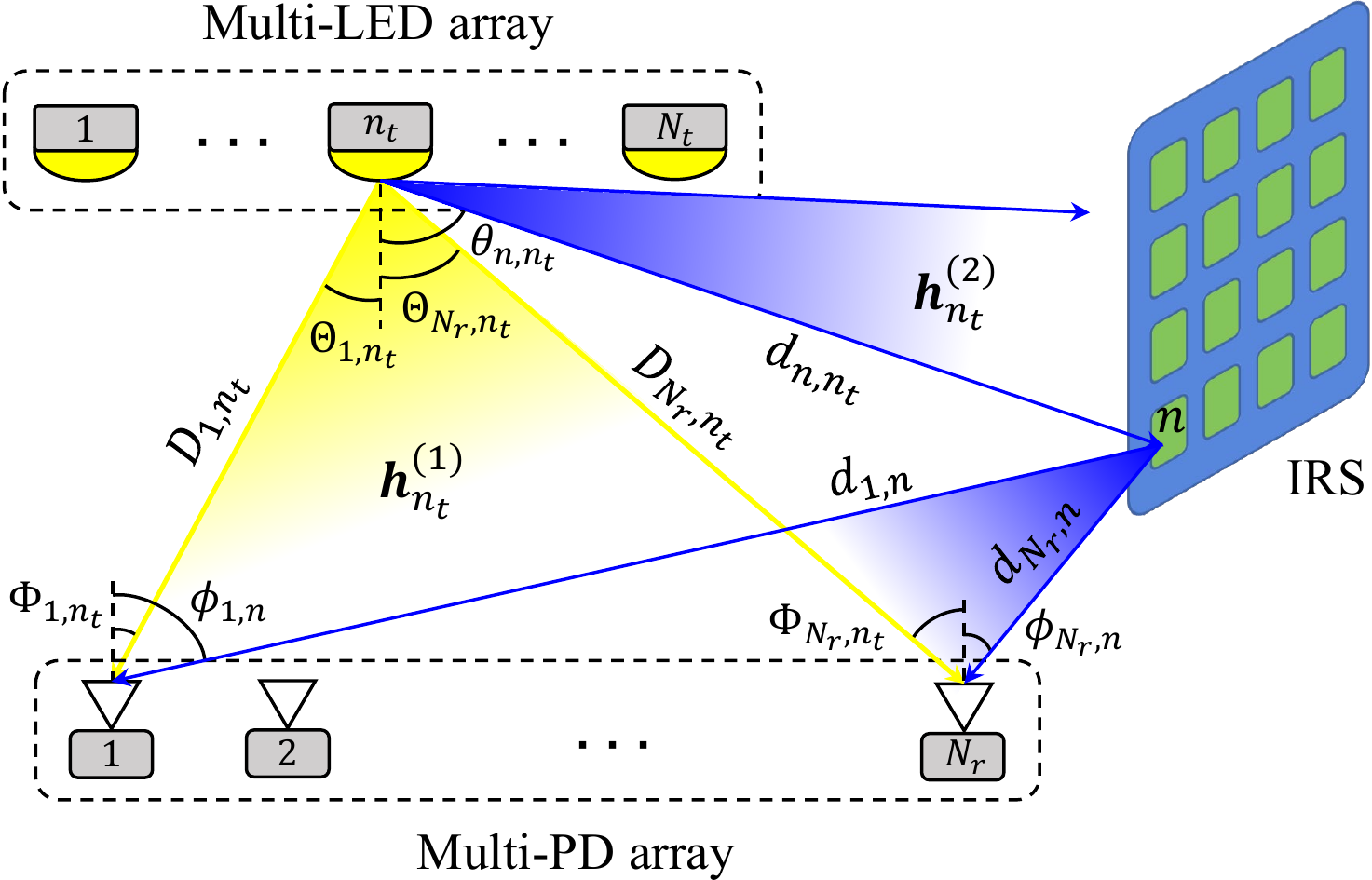}
	\caption{The downlink channel model of IRS-aided MIMO VLC.}
	\vspace{-0.9cm}
	\label{Fig:System}
\end{figure} 

%\vspace{-0.3cm}
\textit{\textbf{LoS channel gain:} }
Generally, the LoS channel gain in VLC can be modeled by the Lambertian radiant formulation, which depends on the locations of transceivers~\cite{pathak2015visible}. 
Thus, the channel gain between the $n_t$-th LED and the $n_r$-th PD is defined as
\begin{equation}
	\label{LoS_VLC_Link}
	h_{n_r,n_t}^{\left(1\right)} = \frac{A_p(m+1)g_{of}}{2\pi D_{n_r,n_t}^2}\cos^m(\Theta_{n_r,n_t})\cos(\Phi_{n_r,n_t})f(\Phi_{n_r,n_t}),
\end{equation}
where parameters $A_p$, $m$, and $g_{of}$ denote the PD area, the Lambertian index, and the optical filter gain, respectively~\cite{pathak2015visible}.
The function $f(\cdot)$ represents the optical concentrator gain, which can be modeled as~\cite{obeed2019optimizing}
\begin{equation}
	\label{Eq:Concentra_gain}
	f(\Phi) = 
	\left\{
	\begin{aligned}
		&\frac{q^2}{\sin^2(\Phi_0)},& \text{if} & \  0 \leq \Phi \leq \Phi_0, \\
		&0, 					 & \text{if} & \  \Phi > \Phi_0,
	\end{aligned}
	\right.
\end{equation}
where $q$ is the refractive index and $\Phi_0$ represents the semi-angle of the FoV.

\textit{\textbf{NLoS channel gain:}}
The NLoS paths in VLC can be categorized into two types: the diffuse reflection path and the specular reflection path.
The channel gain of the former one has been established as a ``multiplicative'' model~\cite{komine2004fundamental}, and experimental results indicate that even the strongest diffuse component is more than 7 dB lower than the weakest LoS component~\cite{ying2015joint}.
Therefore, only specular reflection paths are considered as the NLoS component in this paper.

Generally, wireless channel characteristics depend on the Rayleigh distance $d_0$, which is given by~\cite{tang2021wireless}
\begin{equation}
	d_0 = \frac{2L^2}{\lambda_c},
\end{equation}
where $L$ and $\lambda_c$ denote the dimension of the antenna array and the carrier wavelength, respectively.
The plane wave is a proper model in the far-field scenario, which corresponds to the case when the propagation distance between transceivers is larger than $d_0$.
Otherwise, the angle of arrival is anisotropy and the spherical wave model should be used in the near-field scenario.

In our considered system, due to the nanoscale wavelength of the visible light, $d_0$ gets hundreds of times larger than that in mmWave communications and thus the effect of the near-field becomes more dominant in the optical IRS-aided VLC.
In such an extremely near-field scenario, it has been proven that the Gaussian beam is a reasonable model, and geometric optics can be utilized to analyze the reflected beam~\cite{najafi2019intelligent,2021intelligent}.
In~\cite{abdelhady2020visible,9681888}, optical channel impulse responses of the mirror array-based design and the meta-surface-based design are derived, and the temporal characteristics of the reflected channel are analyzed using the radiometric concept.
Moreover, the phase discontinuities required for the metasurface design and the mirror array design are also obtained.
In this paper, the channel gain of the NLoS link between the $n_t$-th LED and the $n_r$-th PD, which is specularly reflected by the $n$-th IRS unit, is defined as $h_{n_r,n,n_t}^{\left(2\right)}$.
Under the point source assumption, it has been shown that $h_{n_r,n,n_t}^{\left(2\right)}$ can be approximated as~\cite{abdelhady2020visible}
\begin{equation}
	\label{NonLoS_VLC_Link}
	h_{n_r,n,n_t}^{\left(2\right)}=\gamma\frac{A_p(m+1)\cos^m(\theta_{n,n_t})}{2\pi (d_{n,n_t}+d_{n_r,n})^2}g_{of}\cos(\phi_{n_r,n})f(\phi_{n_r,n}),
\end{equation}
where $\gamma$ represents the reflectivity of the IRS unit.
It is observed that the NLoS channel gain in~(\ref{NonLoS_VLC_Link}) is inversely proportional to the square of the distance-sum, showing that the reflected beam can be regarded as emitting from an image transmitter according to geometric optics~\cite{2021intelligent}.

%\vspace{-0.4cm}
\subsection{IRS-aided MIMO VLC Channel}
\label{Subsec:MIMO_channel}
The channel gains between an $N_t$-LED transmitter and an $N_r$-PD receiver are denoted by a matrix $\boldsymbol{H} \in \mathbb{R}^{N_r\times N_t}_+$, which is comprised of an LoS MIMO channel matrix $\boldsymbol{H}_1 \triangleq [\boldsymbol{h}_1^{(1)}, \cdots, \boldsymbol{h}_{N_t}^{(1)}] \in \mathbb{R}^{N_r\times N_t}_+$ and an NLoS MIMO channel matrix $\boldsymbol{H}_2 \triangleq [\boldsymbol{h}_1^{(2)}, \cdots, \boldsymbol{h}_{N_t}^{(2)}] \in \mathbb{R}^{N_r\times N_t}_+$ as
\begin{equation}
	\label{Eq:H_H1_H2}
	\boldsymbol{H} = \boldsymbol{H}_1 + \boldsymbol{H}_2\left( \boldsymbol{F}, \boldsymbol{G} \right),
\end{equation}
where $\boldsymbol{H}_2$ is dependent on IRS configuration matrices $\boldsymbol{F}$ and $\boldsymbol{G}$, $\boldsymbol{h}_{n_t}^{(1)}$ and $\boldsymbol{h}_{n_t}^{(2)}$ denote the LoS and NLoS spatial signatures of the $n_t$-th LED, respectively.
The generation of $\boldsymbol{H}_1$ can be accomplished by aggregating $N_tN_r$ single antenna channel gains, where the element in the $n_r$-th row and the $n_t$-th column is given by~(\ref{LoS_VLC_Link}), i.e., $\boldsymbol{h}_{n_t}^{(1)} \triangleq [h_{1,n_t}^{\left(1\right)}, \cdots, h_{N_r,n_t}^{\left(1\right)}]^T$.

Next, we model the NLoS MIMO channel $\boldsymbol{H}_2$. 
Specifically, the size of optical IRS can be considered infinitely large as compared to the wavelength of the light, which ensures the extremely near-field assumption, i.e., the signal propagation is considered to conform to geometric optics~\cite{2021intelligent}.
As a result, the reflected beam by optical IRS can be considered as emitted from the image source~\cite{9627820}, and its direction follows Snell's law of reflection.
Also, it has been verified in~\cite{9627820} that the power density distribution of such reflected beam in VLC is far more concentrated than that in mmWave and THz bands.
Therefore, when a beam emitted from the $n_t$-th LED and reflected by the $n$-th IRS unit arrives at the $n_r$-th PD, its leaked power to other PDs at sufficiently far-apart locations is practically negligible and thus can be ignored.
Inspired by this, a binary matrix is introduced in~\cite{sun2022joint} to characterize the association behavior between optical IRS units and LESs, and the IRS configuration problem is transformed into an assignment problem.

In our considered MIMO setup, two association matrices are defined as $\boldsymbol{G} \triangleq [ \boldsymbol{g}_1, ..., \boldsymbol{g}_{N_t} ] \in\{0,1\}^{N\times N_t}$ and $\boldsymbol{F} \triangleq [ \boldsymbol{f}_1, ..., \boldsymbol{f}_{N_r} ] \in\{0,1\}^{N\times N_r}$ given in~(\ref{Eq:H_H1_H2}), 
which denote the association behavior between IRS units and different LEDs/PDs, respectively.
More specifically, $g_{n, n_t} = 1$ and $f_{n, n_r} = 1$ indicate that the $n$-th unit is assigned to the $n_t$-th LED and the $n_r$-th PD, i.e., the emission light from this LED is reflected by the $n$-th unit and propagates to the $n_r$-th PD.
Based on the aforementioned definitions, the NLoS MIMO channel matrix of the IRS-aided VLC system can be expressed as
\begin{align}
	\label{Eq:H2_generation}
	\boldsymbol{H}_2 &= \left[
	\begin{matrix}
		\left(\boldsymbol{f}_1 \odot \boldsymbol{g}_1\right)^T\boldsymbol{h}_{1,1}^{(2)}  &  \cdots  &  \left(\boldsymbol{f}_1 \odot \boldsymbol{g}_{N_t}\right)^T\boldsymbol{h}_{1,N_t}^{(2)} \\
		\vdots  &  \ddots  &  \vdots \\
		\left(\boldsymbol{f}_{N_r} \odot \boldsymbol{g}_1\right)^T\boldsymbol{h}_{N_r,1}^{(2)}  &  \cdots  &  \left(\boldsymbol{f}_{N_r} \odot \boldsymbol{g}_{N_t}\right)^T\boldsymbol{h}_{N_r,N_t}^{(2)} \\
	\end{matrix}
	\right],
\end{align}
where the channel gain between the $n_t$-th LED and the $n_r$-th PD is denoted by $\boldsymbol{h}_{n_r,n_t}^{(2)} \triangleq [h_{n_r,1,n_t}^{\left(2\right)}, \cdots, h_{n_r,N,n_t}^{\left(2\right)}]^T$.
Moreover, both $\boldsymbol{G}$ and $\boldsymbol{F}$ have row summation constraints as
\begin{align}
	&\sum_{n_t=1}^{N_t}g_{n,n_t} \leq 1, \quad \forall\ \! n \in \mathcal{N}, \label{Eq:Row summation1}\\
	& \sum_{n_r=1}^{N_r}f_{n, n_r}\leq1, \quad \forall\ \! n \in \mathcal{N}, \label{Eq:Row summation2}
\end{align}
since each IRS unit can be assigned to at most one LED and one PD.

%\begin{figure*}[t]
%	\centering
%	\includegraphics[width=0.9\textwidth]{Blockdiagram.pdf}
%	\caption{The architecture of the proposed IRS-aided MIMO VLC system: baseband digital precoding and detection are adopted at transceivers and IRS is deployed to control the channel gains.}
%	\vspace{-0.5cm}
%	\label{Fig:Block}
%\end{figure*} 

\section{Signal Model and Problem Formulation}
\label{Sec:Model}
In this section, the signal model of IRS-aided MIMO VLC is described in detail, including the modulation scheme, the precoding at the transmitter, and the detection at the receiver.
Then, the MSE metric is derived as the objective of the formulated optimization problem subject to practical VLC constraints.

\subsection{Signal Model of IRS-Aided MIMO VLC}
\label{Subsec:signal}
As depicted in Fig.~\ref{Fig:Block}, an indoor MIMO VLC system with an $N_t$-LED transmitter and an $N_r$-PD receiver is considered, where the MIMO channel consists of the LoS component $\boldsymbol{H}_1$ and the NLoS component $\boldsymbol{H}_2$.
Without loss of generality, the bipolar pulse amplitude modulation (PAM) scheme is adopted to modulate the bit stream, which is split into $N_s$ parallel sub-streams.
In one time slot, $\log_2 M$ independent bits in the $n_s$-th sub-stream are mapped to a symbol $x_{n_s}$, which is expressed as
%\vspace{-0.2cm}
\begin{equation}
	\label{Eq:PAM_define}
	x_{n_s} \in \mathcal{X} : \{\pm I, \pm 3I, \cdots, \pm (M-1)I\}, \ \forall n_s \in \{1,\cdots,N_s\},
\end{equation}
where $I = \sqrt{3/(M^2+1)}$ is a normalization parameter due to $\mathbb{E}(x_{n_s}^2) = \sigma_{x}^2 = 1$.
Then, the generated signal vector $\boldsymbol{x}\in \mathcal{X}^{N_s}$ is precoded by a digital precoding matrix $\boldsymbol{W} \triangleq [\boldsymbol{w}_1^T, \cdots, \boldsymbol{w}_{N_t}^T]^T\in \mathbb{R}^{N_t\times N_s}$, and the coded signal is added by a direct current bias $\boldsymbol{r} = r_0\boldsymbol{1}_{N_\times 1} \in \mathbb{R}^{N_t\times 1}_+$ to meet with the real-valued and nonnegative constraint on signal amplitude.
After electro-optic conversion, the signal travels through the MIMO VLC channel, and the received signal is expressed as
\begin{equation}
	\boldsymbol{y} = \boldsymbol{H}\left(\boldsymbol{W}\boldsymbol{x}+\boldsymbol{r}\right) + \boldsymbol{\omega},
\end{equation}
where $\boldsymbol{\omega}$ is a zero-mean additive white Gaussian noise (AWGN) with the covariance of $\sigma_{\omega}^2\boldsymbol{I}_{N_r}$.

At the receiver, the direct current $\boldsymbol{H}\boldsymbol{r}$ is removed before the signal detection, and the recovered symbol vector can be obtained as 
\begin{equation}
	\label{Eq:MIMO_testing}
	\widetilde{\boldsymbol{x}} = \boldsymbol{Q}\left(\boldsymbol{y}-\boldsymbol{H}\boldsymbol{r}\right),
\end{equation}
where $\boldsymbol{Q} \in \mathbb{R}^{N_s\times N_r}$ represents the detection matrix.
Furthermore, $\widetilde{\boldsymbol{x}}$ is demodulated as $N_s$ sub-streams and the received bits are rearranged by a parallel-to-serial module.

\begin{figure*}[t]
	\centering
	\includegraphics[width=0.9\textwidth]{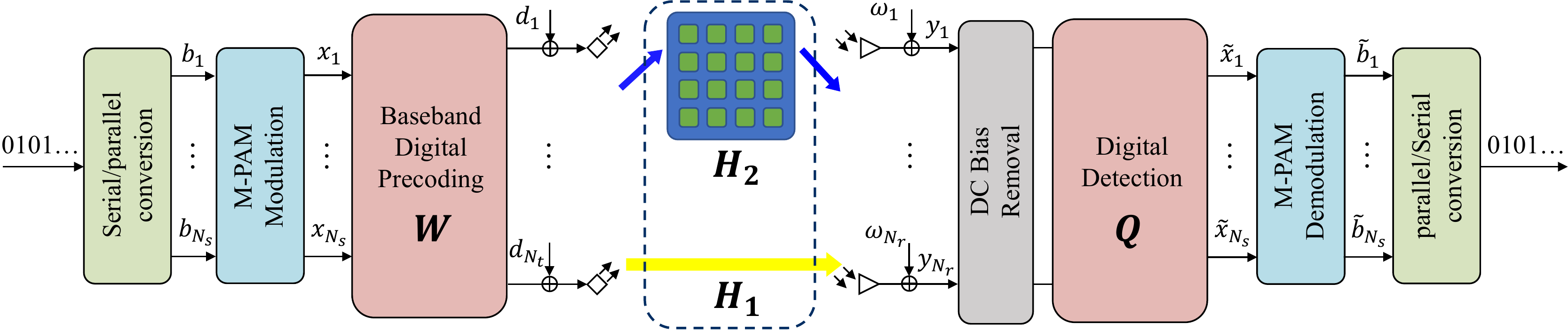}
	\caption{The architecture of the proposed IRS-aided MIMO VLC system: baseband digital precoding and detection are adopted at transceivers and IRS is deployed to control the channel gains.}
	\vspace{-0.5cm}
	\label{Fig:Block}
\end{figure*} 

\subsection{Problem Formulation}
\label{Subsec:Formulation}
This paper aims to minimize the overall MSE of demodulated signals by jointly designing the IRS configuration involved in $\boldsymbol{F}$ and $\boldsymbol{G}$, the precoding matrix $\boldsymbol{W}$ at the transmitter, and the detection $\boldsymbol{Q}$ matrix at the receiver.
First, the MSE metric given by $\text{MSE}( \boldsymbol{F}, \boldsymbol{G}, \boldsymbol{W}, \boldsymbol{Q}) = \mathbb{E}[\left\| \widetilde{\boldsymbol{x}}-\boldsymbol{x} \right\|_2^2]$ can be expanded as
\begin{align}
	\label{Eq:MSE_define}
	\mathbb{E}\left[\left\| \widetilde{\boldsymbol{x}}-\boldsymbol{x} \right\|_2^2\right] & = \mathbb{E}\left[\left\| \left(\boldsymbol{Q}\boldsymbol{H}\boldsymbol{W}-\boldsymbol{I}_{N_s}\right)\boldsymbol{x}+\boldsymbol{Q}\boldsymbol{w} \right\|_2^2\right]\notag\\
	&= \mathbb{E}\left[ \boldsymbol{x}^T\left(\boldsymbol{Q}\boldsymbol{H}\boldsymbol{W}-\boldsymbol{I}_{N_s}\right)^T\left(\boldsymbol{Q}\boldsymbol{H}\boldsymbol{W}-\boldsymbol{I}_{N_s}\right)\boldsymbol{x} \right]  \notag\\
	&\quad + 2\mathbb{E}\left[ \boldsymbol{\omega}^T\boldsymbol{Q}^T\left(\boldsymbol{Q}\boldsymbol{H}\boldsymbol{W}-\boldsymbol{I}_{N_s}\right)\boldsymbol{x} \right] + \mathbb{E}\left( \boldsymbol{\omega}^T\boldsymbol{Q}^T\boldsymbol{Q}\boldsymbol{\omega} \right),
\end{align}
where the last term can be rewritten as
\begin{align}
	\setlength\abovedisplayskip{3pt}
	\mathbb{E}\left( \boldsymbol{\omega}^T\boldsymbol{Q}^T\boldsymbol{Q}\boldsymbol{\omega} \right) = \mathbb{E}\left[\text{tr}(\boldsymbol{\omega}^T\boldsymbol{Q}^T\boldsymbol{Q}\boldsymbol{\omega})\right] = \text{tr}\left( \boldsymbol{Q}\boldsymbol{R}_{\boldsymbol{\omega}}\boldsymbol{Q}^T \right),
	\setlength\belowdisplayskip{3pt}
\end{align}
due to the fact that $\text{tr}(\boldsymbol{C}\boldsymbol{D}) = \text{tr}(\boldsymbol{D}\boldsymbol{C})$.
Following the same process, the MSE in~(\ref{Eq:MSE_define}) can be reformulated as
\begin{align}
	\label{Eq:MSE_final}
	\setlength\abovedisplayskip{3pt}
	\text{MSE}( \boldsymbol{F}, \boldsymbol{G}, \boldsymbol{W}, \boldsymbol{Q}) = \text{tr}\left( \boldsymbol{Q}\boldsymbol{H}\boldsymbol{W} \boldsymbol{R}_{\boldsymbol{x}} (\boldsymbol{Q}\boldsymbol{H}\boldsymbol{W})^T \right) + \text{tr}\left( \boldsymbol{Q}\boldsymbol{R}_{\boldsymbol{\omega}}\boldsymbol{Q}^T \right) +  \text{tr}\left( \boldsymbol{R}_{\boldsymbol{x}} \right) - 2\text{tr}\left( \boldsymbol{Q}\boldsymbol{H}\boldsymbol{W}\boldsymbol{R}_{\boldsymbol{x}} \right),
	\setlength\belowdisplayskip{3pt}
\end{align}
where the signal correlation matrix and the noise correlation matrix are denoted by $\boldsymbol{R}_{\boldsymbol{x}} = \mathbb{E}[\boldsymbol{x}\boldsymbol{x}^T] = \sigma_{x}^2\boldsymbol{I}_{N_s}$ and $\boldsymbol{R}_{\boldsymbol{\omega}} = \mathbb{E}[\boldsymbol{\omega}\boldsymbol{\omega}^T] = \sigma_{\omega}^2\boldsymbol{I}_{N_r}$, respectively.

Due to the practical consideration that VLC system needs to cater to both communication and illumination performance, the transmit power constraint is of great importance in VLC~\cite{wang2020dimming}.
Specifically, the emission power cannot be too large for the comfort of users' eyes, which imposes a constraint at the transmitter as
\begin{equation}
	\label{P:sum_power}
	\setlength\abovedisplayskip{3pt}
	\mathbb{E}\left[\left\| \boldsymbol{W}\boldsymbol{x}+\boldsymbol{r} \right\|_2^2 \right] \leq P_\text{total},
	\setlength\belowdisplayskip{3pt}
\end{equation}
where $P_\text{total}$ is a constant determined by the illumination level.
On the other hand, the real-valued and nonnegative constraint of VLC requires that
\begin{equation}
	\label{P:single_power}
	\boldsymbol{W}\boldsymbol{x}+\boldsymbol{r} \geq \boldsymbol{0},\quad \forall\ \!  \boldsymbol{x} \in \mathcal{X}^{N_s},
\end{equation}
where $\geq$ here means that each element of $\boldsymbol{W}\boldsymbol{x}+\boldsymbol{r}$ is no smaller than zero.
Based on the above discussions, the MSE minimization problem can be formulated as
\begin{align}
	\boldsymbol{P:}\ \min\limits_{\begin{subarray}{c}
			(\ref{P:sum_power}),(\ref{P:single_power})
	\end{subarray}} &\ \text{MSE}\left(\boldsymbol{F}, \boldsymbol{G}, \boldsymbol{W}, \boldsymbol{Q}\right)  \label{P:Object}\\
	\text{s.t.}\ 
	& \sum_{n_t=1}^{N_t}g_{n,n_t} \leq 1, \quad \forall\ \! n \in \mathcal{N}, \label{P:g_row}\\
	& \sum_{n_r=1}^{N_r}f_{n, n_r}\leq1, \quad \forall\ \! n \in \mathcal{N}, \label{P:f_row}\\
	& f_{n,n_r},\ g_{n,n_t}\in\{0,1\}, \quad \forall\ \! n_t \in \mathcal{T},\ n_r \in \mathcal{R}, \label{P:dis}	
\end{align}
where the constraints in~(\ref{P:sum_power})-(\ref{P:single_power}) come from the power limitations and the constraints in~(\ref{P:g_row})-(\ref{P:dis}) are due to the definitions of $\boldsymbol{F}$ and $\boldsymbol{G}$.

\section{Proposed Algorithms to Minimize MSE}
\label{Sec:ProposedAlgorithm_minimizeMSE}
The problem $\boldsymbol{P}$ is a combinational optimization problem and hence is non-deterministic polynomial (NP)-hard to solve.
In this section, an AO algorithm is proposed to optimize the relaxed form of $\boldsymbol{P}$, where the IRS configuration matrices $\boldsymbol{F}$ and $\boldsymbol{G}$, the precoding matrix $\boldsymbol{W}$, and the detection matrix $\boldsymbol{Q}$ are iteratively optimized to minimize MSE.

\subsection{Optimization of IRS Configuration}
When the precoding matrix and the detection matrix are given by $\boldsymbol{W}^{(t)}$ and $\boldsymbol{Q}^{(t)}$ in the $t$-th iteration and kept fixed, the MSE is dependent on IRS configuration matrices $\boldsymbol{F}$ and $\boldsymbol{G}$ only.
The original problem is thus simplified into the following form
\begin{align}
	\boldsymbol{P1:}\ \min\limits_{\begin{subarray}{c}
			\boldsymbol{F}, \boldsymbol{G}:(\ref{P:g_row}),(\ref{P:f_row}),(\ref{P:dis})
	\end{subarray}}&\ \text{MSE}(\boldsymbol{F}, \boldsymbol{G}, \boldsymbol{W}^{(t)}, \boldsymbol{Q}^{(t)}), \label{P1:Object}
	%		\text{s.t.}\ & \boldsymbol{H} = \boldsymbol{H}_1 + \boldsymbol{H}_2\left( \boldsymbol{F}, \boldsymbol{G} \right), \label{P1:H}\\
	%		& \sum_{n_t=1}^{N_t}g_{n,n_t} \leq 1, \quad \forall\ \! n \in \mathcal{N}, \label{P1:g_row}\\
	%		& \sum_{n_r=1}^{N_r}f_{n, n_r}\leq1, \quad \forall\ \! n \in \mathcal{N}, \label{P1:f_row}\\
	%		& f_{n,n_r},\ g_{n,n_t}\in\{0,1\}, \quad \forall\ \! n_t \in \mathcal{T},\ n_r \in \mathcal{R}. \label{P1:discrete}
\end{align}
which is an integer optimization problem and can hardly be optimized in polynomial time.
However, the orthogonality among columns of $\boldsymbol{f}_{n_r}$ and $\boldsymbol{g}_{n_t}$ makes it feasible to transform $\boldsymbol{P1}$ into a more tractable form as
\begin{align}
	\boldsymbol{P1-a:}\ \min\limits_{\boldsymbol{V}}&\ \text{MSE}(\boldsymbol{V}, \boldsymbol{W}^{(t)}, \boldsymbol{Q}^{(t)}) \label{P1-a:Object}\\
	\text{s.t.}\ 
	& \sum_{p=1}^{N_tN_r}v_{n,p} \leq 1, \quad \forall\ \!  n \in \mathcal{N}, \label{P1-a:v_row}\\
	& v_{n,p}\in\{0,1\}, \quad \forall\ \!  p \in \mathcal{P}, \label{P1-a:discrete}
\end{align}
where $\mathcal{P} \triangleq \{1,2,\cdots,N_tN_r \}$ denotes an index set with the cardinality of $|\mathcal{P}| = N_tN_r$ and $\boldsymbol{V} \triangleq \left[ \boldsymbol{v}_1, \cdots, \boldsymbol{v}_{N_tN_r} \right] \in\{0,1\}^{N\times N_tN_r}$ is a defined matrix with each column given by
\begin{equation}
	\label{Eq:v_vector}
	\boldsymbol{v}_{n_r + (n_t - 1)N_r} = \boldsymbol{f}_{n_r} \odot \boldsymbol{g}_{n_t}, \quad \forall\ \! n_t \in \mathcal{T},\ n_r \in \mathcal{R}.
\end{equation}

\begin{proposition}
	\label{Proposition:problem_transform_equivalent}
	The problem $\boldsymbol{P1-a}$ is equivalent to $\boldsymbol{P1}$.
\end{proposition}

\begin{proof}
	Based on the definition of $\boldsymbol{F}$, the support set of $\boldsymbol{f}_{n_r}$ is disjoint to those of $\boldsymbol{f}_{i}, \forall i \neq n_r$, which reveals the column orthogonality of $\boldsymbol{F}$.
	Similarly, the orthogonality among columns of $\boldsymbol{g}_{n_t}$ can also be guaranteed, and the corresponding properties are formulated as
	\begin{align}
		\boldsymbol{f}_{n_{r1}}^T\boldsymbol{f}_{n_{r2}} = \delta_{n_{r1}-n_{r2}},\quad \boldsymbol{g}_{n_{t1}}^T\boldsymbol{g}_{n_{t2}} = \delta_{n_{t1}-n_{t2}},
	\end{align}
	where $\delta_n$ is the discrete Dirac function~\cite{pathak2015visible}. 
	Therefore, the inner product of two arbitrary vectors produced by Hadamard products of $\boldsymbol{f}_{n_{r}}$ and $\boldsymbol{g}_{n_{t}}$ is given by
	\begin{equation}
		\label{Eq:v_orthogonal}
		\left( \boldsymbol{f}_{n_{r1}} \odot \boldsymbol{g}_{n_{t1}} \right)^T \left( \boldsymbol{f}_{n_{r2}} \odot \boldsymbol{g}_{n_{t2}} \right) = \delta_{n_{r1}-n_{r2}}\delta_{n_{t1}-n_{t2}},
	\end{equation}
	which is nonzero if and only if $n_{r1} = n_{r2}$ and $n_{t1} = n_{t2}$.
	According to~(\ref{Eq:v_vector}) and~(\ref{Eq:v_orthogonal}), the support set of $\boldsymbol{v}_{p}$ is disjoint to that of  $\boldsymbol{v}_{i}, \forall i\in \mathcal{P}, i \neq p$, assuring that row vectors of $\boldsymbol{V}$ are one-hot vectors and thus $\boldsymbol{V}$ has orthogonal columns.

	As for the sufficiency, the ones' locations of $\boldsymbol{F}$ and $\boldsymbol{G}$ can be inferred directly from given $\boldsymbol{V}$.
	Specifically, $v_{n,n_r + (n_t - 1)N_r} = 1$ implies that $f_{n, n_r} = g_{n, n_t} = 1$ and $f_{n, i} = g_{n, j} = 0, \forall i \neq n_r, j \neq n_t$, and the proof is thus completed.
\end{proof}

Therefore, problem $\boldsymbol{P1}$ can be solved by first optimizing $\boldsymbol{P1-a}$ and then recovering the IRS configuration matrices $\boldsymbol{F}$ and $\boldsymbol{G}$ according to Proposition~\ref{Proposition:problem_transform_equivalent}.
In general, heuristic algorithms such as the Tabu search algorithm are effective to optimize $\boldsymbol{P1-a}$.
Nevertheless, in consideration of the computational complexity, a more practical solution is via optimizing the relaxed form of $\boldsymbol{P1-a}$, after which the binary result can be obtained according to the minimum distance projection criterion.
In the sequel, the relaxed problem of $\boldsymbol{P1-a}$ is defined as
\begin{align}
	\boldsymbol{P1-b:}\ \min\limits_{\begin{subarray}{c}
			\boldsymbol{V}:(\ref{P1-a:v_row})
	\end{subarray}}&\ \text{MSE}(\boldsymbol{V}, \boldsymbol{W}^{(t)}, \boldsymbol{Q}^{(t)}) \label{P1-b:Object}\\
	\text{s.t.}\ 
	& 0 \leq v_{n,p} \leq 1, \quad \forall\ \!  n \in \mathcal{N},\ p \in \mathcal{P}. \label{P1-b:discrete}
\end{align}

\begin{lemma}
	\label{Lemma:MSE_convex_H}
	A trace equation holds as
	$\text{tr}\{\boldsymbol{M}^T\boldsymbol{C}\boldsymbol{M}\boldsymbol{D}\}=\text{vec}(\boldsymbol{M})^T(\boldsymbol{D}^T\otimes\boldsymbol{C})\text{vec}(\boldsymbol{M})$~\cite{horn1994topics}.
\end{lemma}

\begin{proposition}
	\label{Proposition:MSE_convex}
	The objective of problem $\boldsymbol{P1-b}$ is convex with respect to $\boldsymbol{V}$.
\end{proposition}

\begin{proof}		
	According to Lemma~\ref{Lemma:MSE_convex_H}, the first term of the MSE function~(\ref{Eq:MSE_final}) is reformulated as
	\begin{align}
		\label{Eq:MSE_quadratic}
		\text{tr}\left( \boldsymbol{Q}^{(t)}\boldsymbol{H}\boldsymbol{W}^{(t)} \boldsymbol{R}_{\boldsymbol{x}} \left(\boldsymbol{Q}^{(t)}\boldsymbol{H}\boldsymbol{W}^{(t)}\right)^T \right) & = \text{tr}\left( \boldsymbol{H}^T\boldsymbol{Q}^{(t)T}\boldsymbol{Q}^{(t)}\boldsymbol{H}\boldsymbol{W}^{(t)}\boldsymbol{R}_{\boldsymbol{x}}\boldsymbol{W}^{(t)T} \right) \notag\\
		&= \text{vec}(\boldsymbol{H})^T\left(\left( \boldsymbol{W}^{(t)}\boldsymbol{R}_{\boldsymbol{x}}\boldsymbol{W}^{(t)T} \right)\otimes\left( \boldsymbol{Q}^{(t)T}\boldsymbol{Q}^{(t)} \right)\right)\text{vec}(\boldsymbol{H}),
	\end{align}
	where the central square matrix can be rewritten based on the properties of Kronecker product as
	\begin{align}
		\label{Eq:MSE_quadratic_proof}
		\left( \boldsymbol{W}^{(t)}\boldsymbol{R}_{\boldsymbol{x}}\boldsymbol{W}^{(t)T} \right)\otimes\left( \boldsymbol{Q}^{(t)T}\boldsymbol{Q}^{(t)} \right) & = \left[ \left( \boldsymbol{W}^{(t)}\boldsymbol{R}_{\boldsymbol{x}}^{\frac{1}{2}} \right)\otimes\boldsymbol{Q}^{(t)T} \right] \left[ \left( \boldsymbol{R}_{\boldsymbol{x}}^{\frac{1}{2}}\boldsymbol{W}^{(t)T} \right)\otimes\boldsymbol{Q}^{(t)} \right]  \notag\\
		&= \left[ \left( \boldsymbol{R}_{\boldsymbol{x}}^{\frac{1}{2}}\boldsymbol{W}^{(t)T} \right)\otimes\boldsymbol{Q}^{(t)} \right]^T  \left[ \left( \boldsymbol{R}_{\boldsymbol{x}}^{\frac{1}{2}}\boldsymbol{W}^{(t)T} \right)\otimes\boldsymbol{Q}^{(t)} \right] \succ \boldsymbol{0},
	\end{align}
	which is a positive definite matrix.
	Therefore, $\text{MSE}$ is convex with respect to $\boldsymbol{H}$ since another $\boldsymbol{H}$-dependent term in~(\ref{Eq:MSE_final}) is affine.
	Define an auxiliary matrix $\boldsymbol{H}^{\textit{NLoS}} \in \mathbb{R}^{N\times N_tN_r}_+$ as
	\begin{equation}
		\label{Eq:H_nlos}
		\boldsymbol{H}^{\textit{NLoS}} \triangleq \left[ \boldsymbol{h}_{1,1}^{(2)}, \cdots, \boldsymbol{h}_{N_r,1}^{(2)}, \boldsymbol{h}_{1,2}^{(2)}, \cdots, \boldsymbol{h}_{N_r,N_t}^{(2)} \right],
	\end{equation}
	and then according to the definition of~(\ref{Eq:H2_generation}), the vectorization of $\boldsymbol{H}_2$ can be expressed as
	\begin{equation}
		\label{Eq:V_linear}
		\text{vec}\left( \boldsymbol{H}_2 \right) = \text{diag}\left( \boldsymbol{H}^{\textit{NLoS}} \right)^T\text{vec}\left( \boldsymbol{V} \right),
	\end{equation}
	which shows that MSE is convex with respect to $\boldsymbol{V}$ since $\boldsymbol{H}$ and $\boldsymbol{V}$ are affine.
\end{proof}

Typically, the penalty method and gradient descent algorithm can be used to optimize a convex optimization problem with ensured global optimality~\cite{Boyd}.
However, the projected gradient descent (PGD) algorithm, which will be elaborated in Section~\ref{Subsec:PGD}, is more efficient compared to the aforementioned algorithms~\cite{bubeck2014convex} in terms of convergence rate.
To summarize, the IRS configuration problem $\boldsymbol{P1}$ is solved by optimizing its equivalent problem $\boldsymbol{P1-a}$, which is relaxed as $\boldsymbol{P1-b}$ and then optimized by the PGD algorithm.

\subsection{Optimizations of MIMO Precoding and Detection}
\label{Subsec:Precoding}
Designs of the precoding matrix $\boldsymbol{W}$ and detection matrix $\boldsymbol{Q}$ are discussed separately in this subsection, after which the overall proposed algorithm is presented to solve $\boldsymbol{P}$.

\textit{\textbf{1) Optimize $\boldsymbol{W}$ given $\boldsymbol{V}^{(t+1)}$ and $\boldsymbol{Q}^{(t)}$:}}
With the IRS configuration matirx $\boldsymbol{V}^{(t+1)}$ and the detection matrix $\boldsymbol{Q}^{(t)}$ given as fixed, the original problem $\boldsymbol{P}$ is transformed into a simplified problem $\boldsymbol{P2}$, which aims to optimize the precoding matrix $\boldsymbol{W}$ with fixed $\boldsymbol{V}^{(t+1)}$ and $\boldsymbol{Q}^{(t)}$.
Specifically, the power constraint in~(\ref{P:sum_power}) is formulated by
\begin{align}
	\label{Eq:total_power_constraint}
	\mathbb{E}\left[\left\| \boldsymbol{W}\boldsymbol{x}+\boldsymbol{r} \right\|_2^2 \right] &= \mathbb{E}\left[ \text{tr}\left(\boldsymbol{x}^T\boldsymbol{W}^T\boldsymbol{W}\boldsymbol{x}\right) + 2\boldsymbol{r}^T\boldsymbol{W}\boldsymbol{x} + \boldsymbol{r}^T\boldsymbol{r} \right] \notag\\
	&= \sigma_{x}^2\left\|\boldsymbol{W}\right\|_F^2 + \boldsymbol{r}^T\boldsymbol{r} \leq P_\text{total},
\end{align}
where the second equality holds since $\text{tr}( \boldsymbol{W}^T\boldsymbol{W} ) = \|\boldsymbol{W}\|_F^2$ and $\mathbb{E}[\boldsymbol{x}] = \boldsymbol{0}_{N_s\times 1}$.
Next, the real-valued and nonnegative amplitude constraint in~(\ref{P:single_power}) leads to
\begin{equation}
	\label{Eq:single_power_constraint}
	\sigma_{x}I(M - 1)\| \boldsymbol{w}_{n_t} \|_1  \leq d_{n_t},\quad \forall\ \!  n_t \in \mathcal{T},
\end{equation}
and thus the optimization problem on the precoding matrix $\boldsymbol{W}$ can be expressed as
\begin{align}
	\boldsymbol{P2:}\ \min\limits_{\begin{subarray}{c}
			\boldsymbol{W}:(\ref{Eq:total_power_constraint}), (\ref{Eq:single_power_constraint})
	\end{subarray}}&\ \text{MSE}(\boldsymbol{V}^{(t+1)}, \boldsymbol{W}, \boldsymbol{Q}^{(t)}), \label{P2:Object}
	%		\text{s.t.}\ & \sigma_{x}^2\left\|\boldsymbol{W}\right\|_F^2 \leq P_\text{total} - \boldsymbol{r}^T\boldsymbol{r}, \label{P2:sum_power}\\
	%		& \sigma_{x}I(M - 1)\| \boldsymbol{w}_{n_t} \|_1 \leq d_{n_t}, \forall\ \!  n_t \in \mathcal{T}, \label{P2:single_power}
\end{align}
where the objective is convex with respect to $\boldsymbol{W}$ due to Lemma~\ref{Lemma:MSE_convex_H}.
As for the constraints in~(\ref{Eq:total_power_constraint}) and~(\ref{Eq:single_power_constraint}), the former set is a convex region since it is the sum of the Frobenius norm and a linear term, while the latter one is actually a region bounded by an affine surface.
Consequently, $\boldsymbol{P2}$ is also a convex optimization problem, which can be solved by the PGD algorithm described in Section~\ref{Subsec:PGD}.

\textit{\textbf{2) Optimize $\boldsymbol{Q}$ given $\boldsymbol{V}^{(t+1)}$ and $\boldsymbol{W}^{(t+1)}$:}}
With the IRS configuration matrix $\boldsymbol{V}^{(t+1)}$ and the precoding matrix $\boldsymbol{W}^{(t+1)}$ given as fixed, the problem of optimizing the detection matrix $\boldsymbol{Q}$ turns to be an unconstrained convex optimization problem given by
\begin{equation}
	\label{Pro3}
	\boldsymbol{P3:}\ \min\limits_{\boldsymbol{Q}}\ \text{MSE}(\boldsymbol{V}^{(t+1)}, \boldsymbol{W}^{(t+1)}, \boldsymbol{Q}),
\end{equation}
whose globally optimal point can be attained by setting $\nabla_{\boldsymbol{Q}}\ \text{MSE}(\boldsymbol{V}^{(t+1)}, \boldsymbol{W}^{(t+1)}, \boldsymbol{Q}) = \boldsymbol{0}$,
%\begin{equation}
%	\label{Eq:Pro3_partialDeri}
%	\nabla_{\boldsymbol{Q}}\ \text{MSE}(\boldsymbol{V}^{(t+1)}, \boldsymbol{W}^{(t+1)}, \boldsymbol{Q}) = \boldsymbol{0},
%\end{equation}
which has been derived as~\cite{ying2015joint}
\begin{align}
	\label{Eq:Pro3_resultB}
	\boldsymbol{Q}^{(t + 1)} = \boldsymbol{R}_{\boldsymbol{x}}\left( \boldsymbol{H}^{(t+1)}\boldsymbol{W}^{(t+1)} \right)^T \left( \boldsymbol{H}^{(t+1)}\boldsymbol{W}^{(t+1)}\boldsymbol{R}_{\boldsymbol{x}}\left( \boldsymbol{H}^{(t+1)}\boldsymbol{W}^{(t+1)} \right)^T + \boldsymbol{R}_{\boldsymbol{\omega}} \right)^{-1}
\end{align}
in closed form.

\begin{algorithm}[t]
	\caption{Proposed Alternating Optimization Algorithm to Solve $\boldsymbol{P}$}
	\label{solve_P}
	\textbf{Input:} CSI matrices $\boldsymbol{H}_1$ and $\boldsymbol{H}^{\textit{NLoS}}$, direct current $\boldsymbol{r}$, iteration index $t\gets 0$, and IRS unit index $n\gets 1$.\\
	\textbf{Output:} Precoding matrix $\boldsymbol{W}$, detection matrix $\boldsymbol{Q}$, and IRS configuration matrices $\boldsymbol{F}$, $\boldsymbol{G}$.
	\begin{algorithmic}[1]
		\STATE IRS pre-configuration by minimum distance policy:
		
		$\boldsymbol{G}^{\left(t\right)}$ and $\boldsymbol{F}^{\left(t\right)}$ are initialized by associating each IRS unit to the nearest LED and PD, respectively;
		
		\STATE $\boldsymbol{V}^{(t)}$ is obtained by~(\ref{Eq:v_vector});
		\STATE $\boldsymbol{H}^{(t)}$ is obtained by~(\ref{Eq:H_H1_H2}) and~(\ref{Eq:V_linear});
		\STATE Init: 	
		$\boldsymbol{W}^{(t)} \gets \boldsymbol{H}^{(t)+}$, $\boldsymbol{Q}^{(t)} \gets [\boldsymbol{I}_{N_s}, \boldsymbol{0}_{N_s\times (N_r - N_s)}]$; 
		
		Optimize $\boldsymbol{V}$, $\boldsymbol{W}$, and $\boldsymbol{Q}$ iteratively:
		\REPEAT
		\STATE Solve $\boldsymbol{P1-b}$ to obtain the optimal $\boldsymbol{V}^{(t+1)}$ under the constant parameters $\boldsymbol{W}^{(t)}$ and $\boldsymbol{Q}^{(t)}$; $\!\!\!\!\!\!$
		\STATE Solve $\boldsymbol{P2}$ to obtain the optimal $\boldsymbol{W}^{(t+1)}$ under the constant parameters $\boldsymbol{V}^{(t+1)}$ and $\boldsymbol{Q}^{(t)}$;
		\STATE Solve $\boldsymbol{P3}$ to obtain the optimal $\boldsymbol{Q}^{(t+1)}$ according to the formula~(\ref{Eq:Pro3_resultB}) under the constant parameters $\boldsymbol{V}^{(t+1)}$ and $\boldsymbol{W}^{(t+1)}$;
		\STATE Update the index: $t \gets t + 1$;
		\UNTIL{$\|\text{MSE}^{(t)} - \text{MSE}^{(t-1)}\|\leq \epsilon$}
		
		Recover IRS configuration matrices from continuous $\boldsymbol{V}$:$\!\!\!\!\!\!$
		\REPEAT
		\STATE Find the maximum value: $p^\dagger \gets \arg\max\limits_{p} v_{n,p}$;
		\STATE Obtain associate transceivers by~(\ref{Eq:RecoverFG});
		\STATE Recovery: $f_{n,n_r}\gets 1$ and $f_{n, j\neq n_r} \gets 0$;
		\STATE Recovery: $g_{n,n_t}\gets 1$ and $g_{n, j\neq n_t} \gets 0$;
		\STATE Update the index: $n \gets n + 1$;
		\UNTIL{$n > N$}
	\end{algorithmic}
\end{algorithm}
\setlength{\textfloatsep}{0.5cm}

The overall algorithm to solve the problem $\boldsymbol{P}$ is provided in \textbf{Algorithm}~\ref{solve_P}, which iteratively optimizes the IRS configuration, the precoding matrix at the transmitter, and the detection matrix at the receiver.
Once the MSE converges, IRS configuration matrices $\boldsymbol{F}$ and $\boldsymbol{G}$ can be recovered from the relaxed form of $\boldsymbol{V}$ according to the minimum distance projection criterion.
More specifically, supposing that the column index of the maximum value of the $n$-th row is denoted by $p^\dagger$, its associated transmitter $n_t$ and receiver $n_r$ can be obtained based on Proposition~\ref{Proposition:problem_transform_equivalent} as
\begin{align}
	\label{Eq:RecoverFG}
	n_t = \left\lfloor \frac{p^\dagger-1}{N} \right\rfloor + 1, \quad n_r = \text{mod}\left( p^\dagger-1, N \right) + 1,
\end{align}
%\begin{align}
%	n_t &= \left\lfloor \frac{p^\dagger-1}{N} \right\rfloor + 1, \label{Eq:RecoverG}\\
%	n_r &= \text{mod}\left( p^\dagger-1, N \right) + 1, \label{Eq:RecoverF}
%\end{align}
where $\text{mod}(\cdot)$ denotes the modulus operator and $\lfloor \cdot \rfloor$ is the floor operator.
Then, it can be inferred from~(\ref{Eq:v_vector}) that $f_{n,n_r} = 1$ and $g_{n,n_t} = 1$, while other elements in the $n$-th row of $\boldsymbol{F}$ and $\boldsymbol{G}$ are zeros due to the constraints in~(\ref{P:g_row}) and~(\ref{P:f_row}).

\subsection{PGD Algorithm to Optimize \textbf{P1-b} and \textbf{P2}}
\label{Subsec:PGD}
The PGD algorithm is an effective tool for optimizing convex optimization problems.
Let $\boldsymbol{z}$ and $\boldsymbol{z}^{(i)}$ denote the variable and the fixed variable in the $i$-th iteration, respectively.
The PGD algorithm updates the optimization variables iteratively according to~\cite{bubeck2014convex}
\begin{equation}
	\label{Eq:PGD_projection}
	\setlength\abovedisplayskip{5pt}
	\boldsymbol{z}^{(i+1)} = \arg\min\limits_{\boldsymbol{z}\in \mathcal{Z}}D_h\left( \boldsymbol{z}, \boldsymbol{z}^{(i)} - \alpha\nabla f\left(\boldsymbol{z}^{(i)}\right) \right),
	\setlength\belowdisplayskip{5pt}
\end{equation}
where $\alpha$ denotes the stepsize, $\mathcal{Z}$ is a convex domain of $\boldsymbol{z}$, and $D_h$ indicates a specific distance metric with respect to $\boldsymbol{z}$ and $\boldsymbol{z}^{(i)} - \alpha\nabla f(\boldsymbol{z}^{(i)})$.

\begin{algorithm}[t]
	\caption{IRS Optimization Algorithm to Solve $\boldsymbol{P1-b}$}
	\label{solve_P1}
	\textbf{Input:} CSI matrices $\boldsymbol{H}_1$ and $\boldsymbol{H}^{\textit{NLoS}}$, direct current $\boldsymbol{r}$, iteration index $i\gets 1$.\\
	\textbf{Output:} Relaxed variable $\boldsymbol{V}$.
	\begin{algorithmic}[1]
		\STATE Calculate the minimum point of the Lagrangian function by~(\ref{Eq:dualFunction_derivative_0});
		\STATE The objective function of the dual problem of $\boldsymbol{P1-b}$ is obtained by~(\ref{Eq:LagDual_function});
		\STATE Init: $\boldsymbol{\mu}_1 \gets \epsilon\boldsymbol{1}_{N\times 1}$, $\boldsymbol{\mu}_2 \gets \epsilon\boldsymbol{1}_{NN_tN_r\times 1}$, $\boldsymbol{\mu}_3 \gets \epsilon\boldsymbol{1}_{NN_tN_r\times 1}$;
		
		\REPEAT
		\STATE Calculate $\nabla_{\boldsymbol{\mu}_1}\ g( \boldsymbol{\mu}_1, \boldsymbol{\mu}_2^{(i)}, \boldsymbol{\mu}_3^{(i)} )$ by~(\ref{Eq:dual_mu1_gradient});
		\STATE Gradient Projection: 
		
		$\boldsymbol{\mu}_1^{(i+1)} = \big[ \boldsymbol{\mu}_1^{(i)} + \alpha\nabla_{\boldsymbol{\mu}_1}\ g( \boldsymbol{\mu}_1, \boldsymbol{\mu}_2^{(i)}, \boldsymbol{\mu}_3^{(i)} )|_{\boldsymbol{\mu}_1^{(i)}} \big]^+$;  
		
		\STATE Calculate $\nabla_{\boldsymbol{\mu}_2}\ g( \boldsymbol{\mu}_1^{(i+1)}, \boldsymbol{\mu}_2, \boldsymbol{\mu}_3^{(i)} )$ by~(\ref{Eq:dual_mu2_gradient});
		\STATE Gradient Projection: repeat step 6 for $\boldsymbol{\mu}_2$;
		\STATE Calculate $\nabla_{\boldsymbol{\mu}_3}\ g( \boldsymbol{\mu}_1^{(i+1)}, \boldsymbol{\mu}_2^{(i+1)}, \boldsymbol{\mu}_3 )$ by~(\ref{Eq:dual_mu3_gradient});
		\STATE Gradient Projection: repeat step 6 for $\boldsymbol{\mu}_3$;
		
		\STATE Update the index: $i \gets i + 1$;
		\UNTIL{$\|\text{MSE}^{(i)} - \text{MSE}^{(i-1)}\|\leq \epsilon$}
	\end{algorithmic}
\end{algorithm}
\setlength{\textfloatsep}{0.5cm}

As shown in \textbf{Algorithm}~\ref{solve_P1}, we simplify the projection process by optimizing the dual problem of $\boldsymbol{P1-b}$, whose intrinsic convexity ensures that the dual gap between these two problems is zero~\cite{Boyd}.
The Lagrangian function $\mathcal{L}\left( \boldsymbol{V}, \boldsymbol{\mu}_1, \boldsymbol{\mu}_2, \boldsymbol{\mu}_3 \right)$ of the problem $\boldsymbol{P1-b}$ is given by
%\vspace{-0.1cm}
\begin{align}
	\label{Eq:Lagrangian}
	\setlength\abovedisplayskip{3pt}	
	\mathcal{L} = \text{MSE}\left(\boldsymbol{V}\right) + \boldsymbol{\mu}_1^T\left( \boldsymbol{V}\boldsymbol{1}_{N_tN_r\times 1} - \boldsymbol{1}_{N\times 1} \right) - \boldsymbol{\mu}_2^T\text{vec}\left(\boldsymbol{V}\right) + \boldsymbol{\mu}_3^T\left( \text{vec}\left(\boldsymbol{V}\right) - \boldsymbol{1}_{NN_tN_r\times 1} \right),
\end{align}
\vspace{-0.8cm}

\noindent
where $\boldsymbol{\mu}_1 \in\mathbb{R}_+^{N\times 1} $, $\boldsymbol{\mu}_2 \in\mathbb{R}_+^{NN_tN_r\times 1}$, and $\boldsymbol{\mu}_3 \in\mathbb{R}_+^{NN_tN_r\times 1}$ are Lagrangian multipliers.
Next, the Lagrange dual function can be expressed as~\cite{Boyd}
\begin{equation}
	\label{Eq:LagDual_function}
	g\left( \boldsymbol{\mu}_1, \boldsymbol{\mu}_2, \boldsymbol{\mu}_3 \right) = \mathop{\text{inf}}\limits_{\boldsymbol{V}} \mathcal{L}\left( \boldsymbol{V}, \boldsymbol{\mu}_1, \boldsymbol{\mu}_2, \boldsymbol{\mu}_3 \right),
\end{equation}
where $\mathop{\text{inf}}\limits_{\boldsymbol{V}}$ denotes the infimum of the Lagrangian function with respect to $\boldsymbol{V}$.
The metric $\text{MSE}(\boldsymbol{V})$ is in the quadratic form of $\boldsymbol{V}$ according to Proposition~\ref{Proposition:MSE_convex}, implying that the minimum value of the Lagrangian function can be directly obtained by solving $\nabla_{\boldsymbol{V}} \mathcal{L}( \boldsymbol{V}^*, \boldsymbol{\mu}_1, \boldsymbol{\mu}_2, \boldsymbol{\mu}_3 ) = \boldsymbol{0}$, which is detailed in \textbf{Appendix}~\ref{App:dual function}.
Therefore, the dual optimization problem is given by
\begin{align}
	\boldsymbol{P1-c:}\ \max\limits_{\boldsymbol{\mu}_1, \boldsymbol{\mu}_2, \boldsymbol{\mu}_3}&\ g\left( \boldsymbol{\mu}_1, \boldsymbol{\mu}_2, \boldsymbol{\mu}_3 \right) \label{P1-c:Object}\\
	\text{s.t.}\ 
	& \boldsymbol{\mu}_1, \boldsymbol{\mu}_2, \boldsymbol{\mu}_3 \geq \boldsymbol{0}, \label{P1-c:mu_constraint}
%	& \boldsymbol{\mu}_1 \geq \boldsymbol{0}, \label{P1-c:mu1_constraint}\\
%	& \boldsymbol{\mu}_2 \geq \boldsymbol{0}, \label{P1-c:mu2_constraint}\\
%	& \boldsymbol{\mu}_3 \geq \boldsymbol{0}, \label{P1-c:mu3_constraint}
\end{align}
where $g( \boldsymbol{\mu}_1, \boldsymbol{\mu}_2, \boldsymbol{\mu}_3 ) = \mathcal{L}( \boldsymbol{V}^*, \boldsymbol{\mu}_1, \boldsymbol{\mu}_2, \boldsymbol{\mu}_3 )$, and $\boldsymbol{P1-c}$ is a convex optimization problem~\cite{Boyd}.
The gradient ascend policy is utilized to reach the maximum value of $\boldsymbol{P1-c}$, where derivatives with respect to $\boldsymbol{\mu}_1$, $\boldsymbol{\mu}_2$, and $\boldsymbol{\mu}_3$ are given in \textbf{Appendix}~\ref{App:gradient}.
As for the projection, fortunately, the convex domain has been greatly simplified, and the projection point in $i$-th iteration is given by $\!\!\!\!$
\begin{equation}
	\label{Eq:mu1_update}
	\boldsymbol{\mu}_j^{(i+1)} = \left[ \boldsymbol{\mu}_j^{(i)} + \alpha\nabla_{\boldsymbol{\mu}_j}\ g\left( \boldsymbol{\mu}_1, \boldsymbol{\mu}_2, \boldsymbol{\mu}_3 \right)\big|_{\boldsymbol{\mu}_j^{(i)}} \right]^+\!\!\!,\ j \in \{1, 2, 3\},
\end{equation}
where $[\ \cdot\ ]^+ = \max(\ \cdot\ ,\ 0)$ denotes the nonnegative function. 

Similarly, the PGD algorithm can also be applied to solve $\boldsymbol{P2}$.
A key difference between this problem and the former one lies in the property of the constraint~(\ref{Eq:single_power_constraint}).
The $l$-1 norm generally is non-differentiable, whereas its sub-gradient can be derived as
\begin{equation}
	\frac{\partial \| \boldsymbol{w}_{n_t} \|_1}{\partial \boldsymbol{w}_{n_t} } = \text{sgn}\left( \boldsymbol{w}_{n_t}  \right), \quad \forall\ \!  n_t \in \mathcal{T},
\end{equation}
where $\text{sgn}(\cdot)$ denotes the sign of the argument. 
Moreover, another benefit of optimizing the dual problem is the reduction of the number of variables, i.e., only $N_t + 1$ variables are needed instead of $N_tN_s$ of the original problem.

\subsection{Alternative Solutions in High/Low SNR Regimes}
The MSE minimization problem $\boldsymbol{P}$ has been solved by the proposed algorithm, which optimizes the IRS configuration matrix, the precoding matrix, and the detection matrix iteratively.
Nevertheless, two special cases in the high and low SNR regimes are considered particularly, and the asymptotic analysis of the proposed algorithm is provided in this subsection.

\textit{\textbf{1) High SNR regime:}}
The original problem can be simplified in the high-SNR regime, where a more efficient algorithm can be utilized to jointly optimize the precoding and detection matrices.
To this end, subproblems $\boldsymbol{P2}$ and $\boldsymbol{P3}$ are combined into the following problem
\begin{align}
	\boldsymbol{P4:}\ \min\limits_{\begin{subarray}{c}
			\boldsymbol{W}, \boldsymbol{Q}:(\ref{Eq:total_power_constraint}), (\ref{Eq:single_power_constraint})
	\end{subarray}}&\ \lim\limits_{\sigma_{\omega}^2\to 0}\text{MSE}(\boldsymbol{V}^{(t+1)}, \boldsymbol{W}, \boldsymbol{Q}) \label{Pro4:Object}
	%		\text{s.t.}\ & \sigma_{x}^2\left\|\boldsymbol{W}\right\|_F^2 \leq P_\text{total} - \boldsymbol{r}^T\boldsymbol{r}, \label{Pro4:sum_power}\\
	%		& \sigma_{x}I(M - 1)\| \boldsymbol{w}_{n_t} \|_1 \leq d_{n_t}, \forall\ \!  n_t \in \mathcal{T}. \label{Pro4:single_power}
\end{align}
where the objective can be calculated according to~(\ref{Eq:MSE_final}) as
\begin{align}
	\label{Eq:ZF_object}
	\lim\limits_{\sigma_{\omega}^2\to 0}\text{MSE}(\boldsymbol{V}^{(t+1)}, \boldsymbol{W}, \boldsymbol{Q}) =\ \text{tr}\left( \boldsymbol{Q}\boldsymbol{H}^{(t+1)}\boldsymbol{W}(\boldsymbol{Q}\boldsymbol{H}^{(t+1)}\boldsymbol{W})^T \right) - 2\text{tr}\left( \boldsymbol{Q}\boldsymbol{H}^{(t+1)}\boldsymbol{W} \right) + N_s.
\end{align}

The SVD of $\boldsymbol{H}^{(t+1)}$ is carried out as $\boldsymbol{H}^{(t+1)} = \widetilde{\boldsymbol{U}}^{(t+1)} \boldsymbol{\Lambda}^{(t+1)} \widetilde{\boldsymbol{V}}^{(t+1)T}$, where $\boldsymbol{\Lambda}^{(t+1)}$ is a diagonal matrix with decreasing singular values, and $\widetilde{\boldsymbol{U}}^{(t+1)} = [\widetilde{\boldsymbol{u}}_1^{(t+1)}, \cdots, \widetilde{\boldsymbol{u}}_{N_r}^{(t+1)}]^T$ and $\widetilde{\boldsymbol{V}}^{(t+1)}$ are the left singular matrix and right singular matrix, respectively.
Then, a feasible precoding matrix can be expressed as
\begin{equation}
	\label{Eq:JointOptimize_precoding}
	\boldsymbol{W}_{ZF}^{(t+1)} = \zeta \widetilde{\boldsymbol{V}}^{(t+1)} \boldsymbol{\Lambda}^{(t+1)\dagger} [\widetilde{\boldsymbol{u}}_1^{(t+1)}, \cdots, \widetilde{\boldsymbol{u}}_{N_s}^{(t+1)}],
\end{equation}
where $\boldsymbol{\Lambda}^{(t+1)\dagger}$ represents the transpose of $\boldsymbol{\Lambda}^{(t+1)}$ with each element in the reciprocal form and the coefficient $\zeta$ is a constant given by
\begin{align}
	\label{Eq:JointOptimize_min}
	\zeta = \min \Bigg( \sqrt{\frac{P_\text{total} - \boldsymbol{r}^T\boldsymbol{r}}{\sigma_{x}^2 \left\|  \boldsymbol{\Lambda}^{(t+1)\dagger}[\widetilde{\boldsymbol{u}}_1^{(t+1)}, \cdots, \widetilde{\boldsymbol{u}}_{N_r}^{(t+1)}] \right\|_F^2 }},  \frac{\min\limits_{n_t\in \mathcal{T}} d_{n_t} / \sqrt{\sigma_{x}^2}}{I(M-1) \left\| \widetilde{\boldsymbol{V}}^{(t+1)} \boldsymbol{\Lambda}^{(t+1)\dagger} [\widetilde{\boldsymbol{u}}_1^{(t+1)}, \cdots, \widetilde{\boldsymbol{u}}_{N_s}^{(t+1)}] \right\|_r} \Bigg),
\end{align}
which is determined to satisfy the power constraints in~(\ref{Eq:total_power_constraint}) and~(\ref{Eq:single_power_constraint}).
In the meanwhile, the detection matrix $\boldsymbol{Q}$ can be chosen as
\begin{equation}
	\label{Eq:JointOptimize_detection}
	\boldsymbol{Q}_{ZF}^{(t+1)} = \zeta^{-1}[\boldsymbol{I}_{N_s}, \boldsymbol{0}_{N_s\times (N_r - N_s)}].
\end{equation}

\begin{proposition}
	\label{Propostions:}
	In the high-SNR regime, the proposed precoding matrix $\boldsymbol{W}_{ZF}^{(t+1)}$ and the detection matrix $\boldsymbol{Q}_{ZF}^{(t+1)}$ can converge to the optimal solution of $\boldsymbol{P4}$.
\end{proposition}

\begin{proof}
	According to~(\ref{Eq:ZF_object}), the MSE metric can be rewritten as $\text{tr}(\boldsymbol{B}\boldsymbol{B}^T) - 2\text{tr}(\boldsymbol{B}) + N_s$ with $\boldsymbol{B} = \boldsymbol{Q}\boldsymbol{H}^{(t)}\boldsymbol{W}$, which is a convex function with respect to $\boldsymbol{B}$ due to Lemma~\ref{Lemma:MSE_convex_H}. 
	When constraints of~(\ref{Eq:total_power_constraint}) and~(\ref{Eq:single_power_constraint}) are ignored, the minimum MSE can be obtained by directly forcing the derivative of $\boldsymbol{B}$ equal to $\boldsymbol{0}$, i.e., $\boldsymbol{B} - \boldsymbol{I}_{N_s} = \boldsymbol{0}$.
	Considering that $\boldsymbol{W}_{ZF}^{(t+1)}$ is composed of the first $N_s$ columns of the right pseudo-inverse matrix of $\boldsymbol{H}^{(t)}$, it is easy to confirm that $\boldsymbol{Q}_{ZF}^{(t+1)} \boldsymbol{H}^{(t)} \boldsymbol{W}_{ZF}^{(t+1)} = \boldsymbol{I}_{N_s}$.
	Furthermore, emission power constraints can be ensured due to the proper setting of $\zeta$ in~(\ref{Eq:JointOptimize_min}), implying that the proposed precoding matrix $\boldsymbol{W}_{ZF}^{(t+1)}$ and detection matrix $\boldsymbol{Q}_{ZF}^{(t+1)}$ will converge to the optimal solution of $\boldsymbol{P4}$.
\end{proof}

\textit{\textbf{2) Low SNR regime:}}
When the noise power is sufficiently high as compared to the signal power, the optimal detection matrix can be expressed according to~(\ref{Eq:Pro3_resultB}) as
\begin{equation}
	\label{Eq:B_lowSNR}
	\boldsymbol{Q}^{(t+1)} = \boldsymbol{R}_{\boldsymbol{x}}\left( \boldsymbol{H}^{(t+1)}\boldsymbol{W}^{(t+1)} \right)^T  \boldsymbol{R}_{\boldsymbol{\omega}}^{-1} = \frac{\sigma_{x}^2}{\sigma_{\omega}^2} \left( \boldsymbol{H}^{(t+1)}\boldsymbol{W}^{(t+1)} \right)^T.
\end{equation}
Then, the asymptotic analysis of the MSE performance versus the SNR at the receiver is provided in the following, where the SNR is defined as
\begin{equation}
	\label{Eq:}
	\widehat{\text{SNR}} = \frac{\mathbb{E}\left[\left\| \boldsymbol{H}\boldsymbol{W}\boldsymbol{x} \right\|_2^2 \right]}{\mathbb{E}\left[\left\| \boldsymbol{\omega} \right\|_2^2 \right]} = \frac{\sigma_{x}^2\| \boldsymbol{H}\boldsymbol{W} \|_F^2}{\sigma_{\omega}^2N_r}.
\end{equation}

\begin{lemma}
	\label{Lemma:cauthy}
	$\boldsymbol{a}^T\cdot\boldsymbol{b} \leq \lvert \boldsymbol{a} \rvert \cdot \lvert \boldsymbol{b} \rvert$ for any two vectors $\boldsymbol{a}$ and $\boldsymbol{b}$ (Cauchy-Schwarz inequality~\cite{bubeck2014convex}).
\end{lemma}

The normalized MSE, namely $\text{MSE}( \boldsymbol{F}, \boldsymbol{G}, \boldsymbol{W}, \boldsymbol{Q})/\boldsymbol{R}_{\boldsymbol{x}}$, is used to evaluate the MSE performance, which is given by
\begin{align}
	\label{Eq:MSE_lowSNR}
	\overline{\text{MSE}} &= 1 - \left( 2\text{tr}\left( \boldsymbol{Q}\boldsymbol{H}\boldsymbol{W} \right) - \| \boldsymbol{Q}\boldsymbol{H}\boldsymbol{W} \|_F^2 - \frac{\sigma_{\omega}^2}{\sigma_{x}^2}\| \boldsymbol{Q} \|_F^2 \right) \notag\\
	& = 1 - \frac{N_r}{N_s} \left( \frac{\sigma_{x}^2\| \boldsymbol{H}\boldsymbol{W} \|_F^2}{\sigma_{\omega}^2N_r} - \frac{\sigma_{x}^4}{\sigma_{\omega}^4N_r} \| \left( \boldsymbol{H}\boldsymbol{W} \right)^T\boldsymbol{H}\boldsymbol{W} \|_F^2 \right) \notag\\
	& \leq 1 - \frac{N_r}{N_s} \left( \frac{\sigma_{x}^2\| \boldsymbol{H}\boldsymbol{W} \|_F^2}{\sigma_{\omega}^2N_r} - \frac{\sigma_{x}^4}{\sigma_{\omega}^4N_r} \| \boldsymbol{H}\boldsymbol{W} \|_F^4 \right),
\end{align}
where the inequality comes from Lemma~\ref{Lemma:cauthy}.
Therefore, the normalized MSE in the low-SNR regime can be reformulated as
\begin{equation}
	\label{Eq:LowSNR_normalMSE}
	\overline{\text{MSE}} \approx 1 - \frac{N_r}{N_s} \left( \widehat{\text{SNR}} - o\left(\widehat{\text{SNR}}\right) \right),
\end{equation}
where $o(\cdot)$ denotes the higher order quantity.
It can be observed that $\overline{\text{MSE}}$ decreases linearly with the increase of SNR, and the slope is $N_r/N_s$.

\subsection{Theoretical Analysis of Proposed Algorithms}
\label{Subsec:analysis}
The convergence of the proposed AO algorithm is guaranteed due to the convexity of three subproblems.
Specifically, the relaxed IRS configuration matrix, the precoding matrix, and the detection matrix in the $t$-th iteration are represented by $\boldsymbol{V}^{(t)}$, $\boldsymbol{W}^{(t)}$, and $\boldsymbol{Q}^{(t)}$, respectively.
Then in the next iteration, $\boldsymbol{V}^{(t+1)}$ is obtained by optimizing the convex optimization problem $\boldsymbol{P1-b}$ and thereby we have $\text{MSE}(\boldsymbol{V}^{(t+1)}, \boldsymbol{W}^{(t)}, \boldsymbol{Q}^{(t)}) \leq \text{MSE}(\boldsymbol{V}^{(t)}, \boldsymbol{W}^{(t)}, \boldsymbol{Q}^{(t)})$.
Similarly, it can also be confirmed that $\text{MSE}(\boldsymbol{V}^{(t+1)}, \boldsymbol{W}^{(t+1)}, \boldsymbol{Q}^{(t+1)}) \leq \text{MSE}(\boldsymbol{V}^{(t+1)}, \boldsymbol{W}^{(t+1)}, \boldsymbol{Q}^{(t)}) \leq \text{MSE}(\boldsymbol{V}^{(t+1)}, \boldsymbol{W}^{(t)}, \boldsymbol{Q}^{(t)})$ since both $\boldsymbol{P2}$ and $\boldsymbol{P3}$ are convex optimization problems.
Finally, the MSE will converge to a certain nonnegative value since it is lower-bounded by 0.

The complexity of the proposed algorithm depends on the number of iterations in $\boldsymbol{P1}$-$\boldsymbol{P3}$ and the number of matrix operations in each iteration.
Define the number of outer loops of \textbf{Algorithm}~\ref{solve_P} and the numbers of inner loops to problems of $\boldsymbol{P1}$ and $\boldsymbol{P2}$ are $I_{O}$, $I_1$, and $I_2$, respectively.
The standard big-O metric $\mathcal{O}(\cdot)$ is used to evaluate the computational complexity.
According to~\cite{bubeck2014convex}, $I_1$ is in the order of $\mathcal{O}(\log(1/\epsilon))$ with $\epsilon$ denoting the tolerance error since the Hessian matrix of the objective shown in~(\ref{Eq:MSE_quadratic_proof}) is a positive definite matrix.
This complexity is far lower than $\mathcal{O}(1/\epsilon)$ in the gradient descent algorithm~\cite{bubeck2014convex}, showing the superior efficiency of the PGD algorithm.
Moreover, the second-order derivative of the objective in $\boldsymbol{P2}$ is given by
\begin{align}
	\label{Eq:Pro2_MSE_quadratic}
	\nabla_{\text{vec}(\boldsymbol{W})}^2 \text{MSE}(\boldsymbol{V}^{(t+1)}, \boldsymbol{W}, \boldsymbol{Q}^{(t)}) =&\ \boldsymbol{R}_{\boldsymbol{x}} \otimes \left( \boldsymbol{H}^{(t+1)T}\boldsymbol{Q}^{(t)T}\boldsymbol{Q}^{(t)}\boldsymbol{H}^{(t+1)} \right) \notag\\
	=&\ \left( \boldsymbol{R}_{\boldsymbol{x}}^{\frac{1}{2}} \otimes \boldsymbol{Q}^{(t)}\boldsymbol{H}^{(t+1)} \right)^T \left(\boldsymbol{R}_{\boldsymbol{x}}^{\frac{1}{2}} \otimes \boldsymbol{Q}^{(t)}\boldsymbol{H}^{(t+1)} \right) \succ\ \lambda_{\min} \boldsymbol{I}_{N_tN_s},
\end{align}
where $\lambda_{\min}$ is the minimum eigenvalue of the Hessian matrix and $I_2$ is also in the order of $\mathcal{O}(\log(1/\epsilon))$ due to the $\lambda_{\min}$-strong convexity~\cite{bubeck2014convex}.

In each iteration, 9 matrix multiplications and 1 matrix inversion are needed to calculate the Lagrange dual function, and 10 matrix multiplications are used to update Lagrangian multipliers for optimizing $\boldsymbol{P1-b}$.
Then, $\boldsymbol{P2}$ is also optimized by the PGD algorithm, where the number of constraints is reduced from $N^2N_tN_r$ to $N_t + 1$ as compared to $\boldsymbol{P1-b}$.
The optimal solution of $\boldsymbol{P3}$ can be directly obtained from~(\ref{Eq:Pro3_resultB}), which requires 6 matrix multiplications, 1 matrix addition, and 1 matrix inversion.
Finally, the complexity of the overall proposed algorithm is $\mathcal{O}(I_{O}(10I_1 + 6I_2 + 19) + NN_tN_r)$, where the last term is due to the operations for recovering $\boldsymbol{F}$ and $\boldsymbol{G}$ from the optimized $\boldsymbol{V}$.

%****************************************************仿真******************************************************************
\section{Numerical Results}
\label{Sec:Num}
In this section, numerical results are provided to evaluate the performance of the proposed algorithm for minimizing the MSE of the IRS-aided MIMO VLC system.
To this end, the downlink communication of an indoor MIMO VLC system is considered, where the transmitter equipped with $16$ LEDs and the receiver equipped with $4$ PDs are both set as uniform planar arrays (UPAs). 
Suppose the room size is $8$m $\times$ $8$m $\times$ $3$m, the coordinates of LEDs are generated by dividing the ceiling into $16$ equal squares and placing LEDs in the centers of them.
Without loss of generality, the center of the UPA receiver is located at $(2.0 \text{m},\ 3.2 \text{m},\ 1.0 \text{m})$ of the 3D coordinate system, and IRS units are deployed equally within the rectangular determined by corners of $(0.0 \text{m},\ 1.0 \text{m},\ 1.2 \text{m})$ and $(0.0 \text{m},\ 7.0 \text{m},\ 2.9 \text{m})$.
More detailed settings are listed in Table~\ref{Tab:Parameters}, including the VLC channel parameters, the specific modulation scheme, etc.
Since the CSI is assumed known to the transceiver, the channel gain in indoor VLC scenario is assumed to be quasi-static and these parameters remain unchanged in our simulations.
\begin{table*}[t]
	\centering
	\small
	\caption{Simulation parameters}
	\label{parameters}
	\begin{tabular}{| l | l || l | l |}
		\hline
		\multicolumn{1}{|c|}{\textbf{Parameters}} &
		\multicolumn{1}{c||}{\textbf{Values}} & \multicolumn{1}{c|}{\textbf{Parameters}} & \multicolumn{1}{c|}{\textbf{Values}} \\
		\hline
		The number of LEDs, $N_t$ & 16 & The Lambertian index, $m$ & $1$ \\
		The number of PDs, $N_r$ & 4 & The reflectivity of each IRS unit, $\gamma$ & $0.9$ \\
		The number of independent data streams, $N_s$ & 4 & Optical filter gain, $g_{of}$ & $1$ \\
		The number of IRS units, $N$ & 64 & The area of a single PD, $A_p$ & $1\ \text{cm}^2$ \\
		The area of an IRS unit & $10\ \text{cm}^2$  & The PD spacing & $0.2$ m \\
		The operating height of the receiver & $1$ m & Semi-angle of the FoV, $\Phi_0$ & $60\degree$ \\
		The order of bipolar PAM, $M$ & $4$ & Refractive index of the PD, $q$ & $1.5$ \\
		Direct current bias, $\boldsymbol{r}$ & $\boldsymbol{1}_{N_t \times 1}$ & Tolerance error, $\epsilon$ & $10^{-6}$ \\
		Total emission power constraint, $P_{\text{total}}$ & $160$ W & The power of AWGN at the receiver, $\sigma_{\omega}^2$ & $10^{-14}$ \\
		\hline
	\end{tabular}
	\label{Tab:Parameters}
\end{table*}

The performance of the proposed AO algorithm is compared to the following alternative IRS configuration schemes: 
(1) Distance greedy IRS scheme: IRS configuration is carried out by matching each IRS unit and transceivers with the minimum distance, i.e., $f_{n, n_r} = 1$ and $g_{n, n_t} = 1$ denote that the $n$-th IRS unit has the shortest distances with the $n_r$-th PD and the $n_t$-th LED, respectively; 
(2) Random IRS scheme: The binary variable $\boldsymbol{V}$ in $\boldsymbol{P1}$ is randomly generated, and each simulation point is the average result of $5000$ independent experiments;
(3) No IRS scheme: In this scheme, the number of IRS units is 0, and therefore the MIMO VLC channel gain is given as $\boldsymbol{H} = \boldsymbol{H}_1$.
On the other hand, two commonly adopted precoding solutions are considered to show the performance of the proposed precoding scheme: 
(1) ZF precoding scheme: The precoding matrix $\boldsymbol{W}$ is determined by the right pseudo-inverse matrix of the MIMO VLC gain matrix $\boldsymbol{H}$, and the Frobenius norm of $\boldsymbol{W}$ needs to meet the emission power constraints of~(\ref{Eq:total_power_constraint}) and~(\ref{Eq:single_power_constraint});
(2) MMSE precoding scheme: The precoding matrix is jointly determined by $\boldsymbol{H}$ and the AWGN power $\sigma_{\omega}^2$ for maximizing the signal-to-interference-plus-noise ratio (SINR) at the receiver. 
Moreover, the lower bound of MSE can be obtained by solving the relaxed form of $\boldsymbol{P}$, which is considered as the optimal baseline.

\begin{figure}[t]
	\centering
	\begin{minipage}[b]{0.43\linewidth}
		\centering
		\includegraphics[width=1\textwidth]{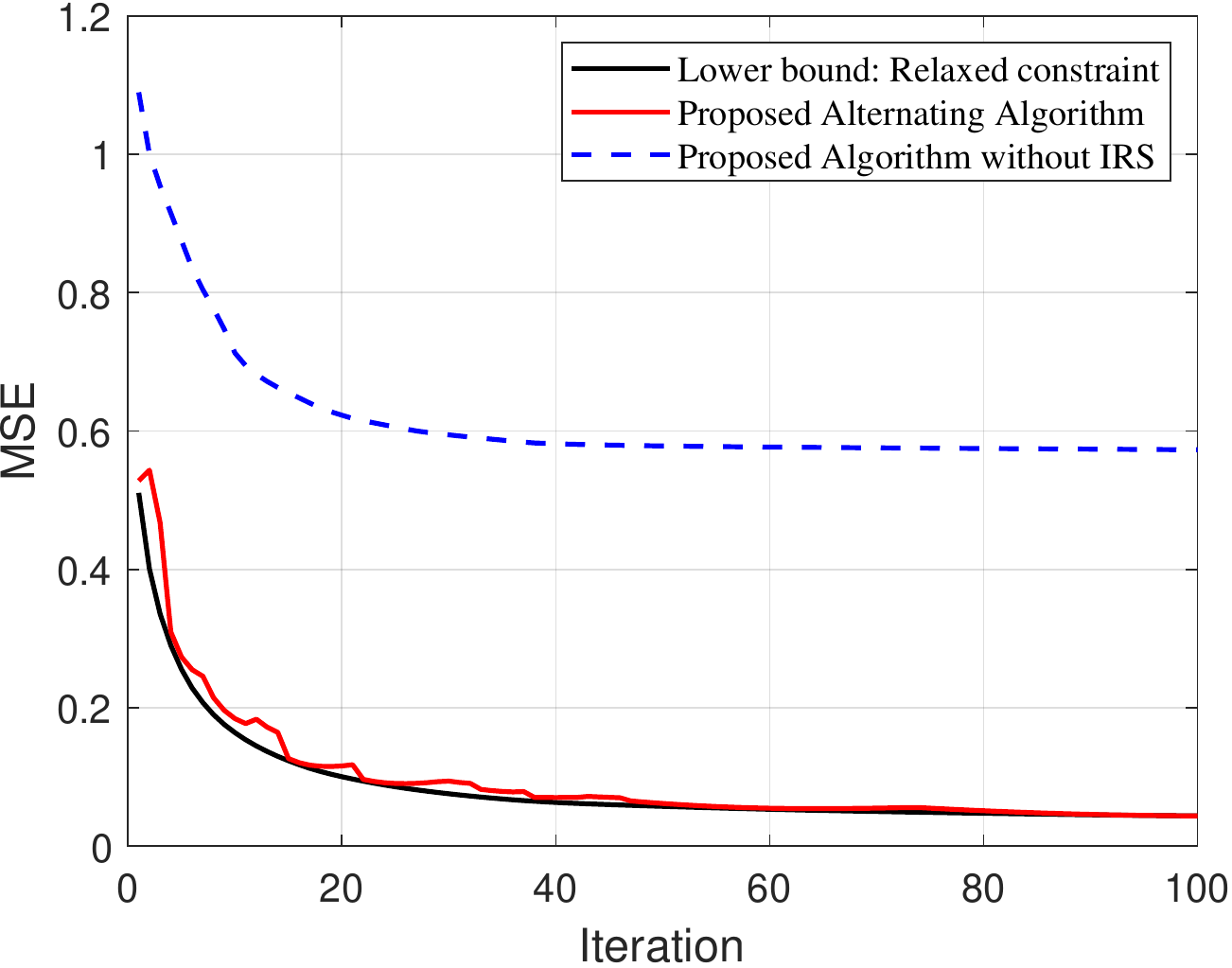}
		\caption{Convergence performance of the proposed AO algorithm.}
		\label{Fig:Iteration}
	\end{minipage}
	\hspace{0.3cm}
	\begin{minipage}[b]{0.43\linewidth}
		\centering
		\includegraphics[width=1\textwidth]{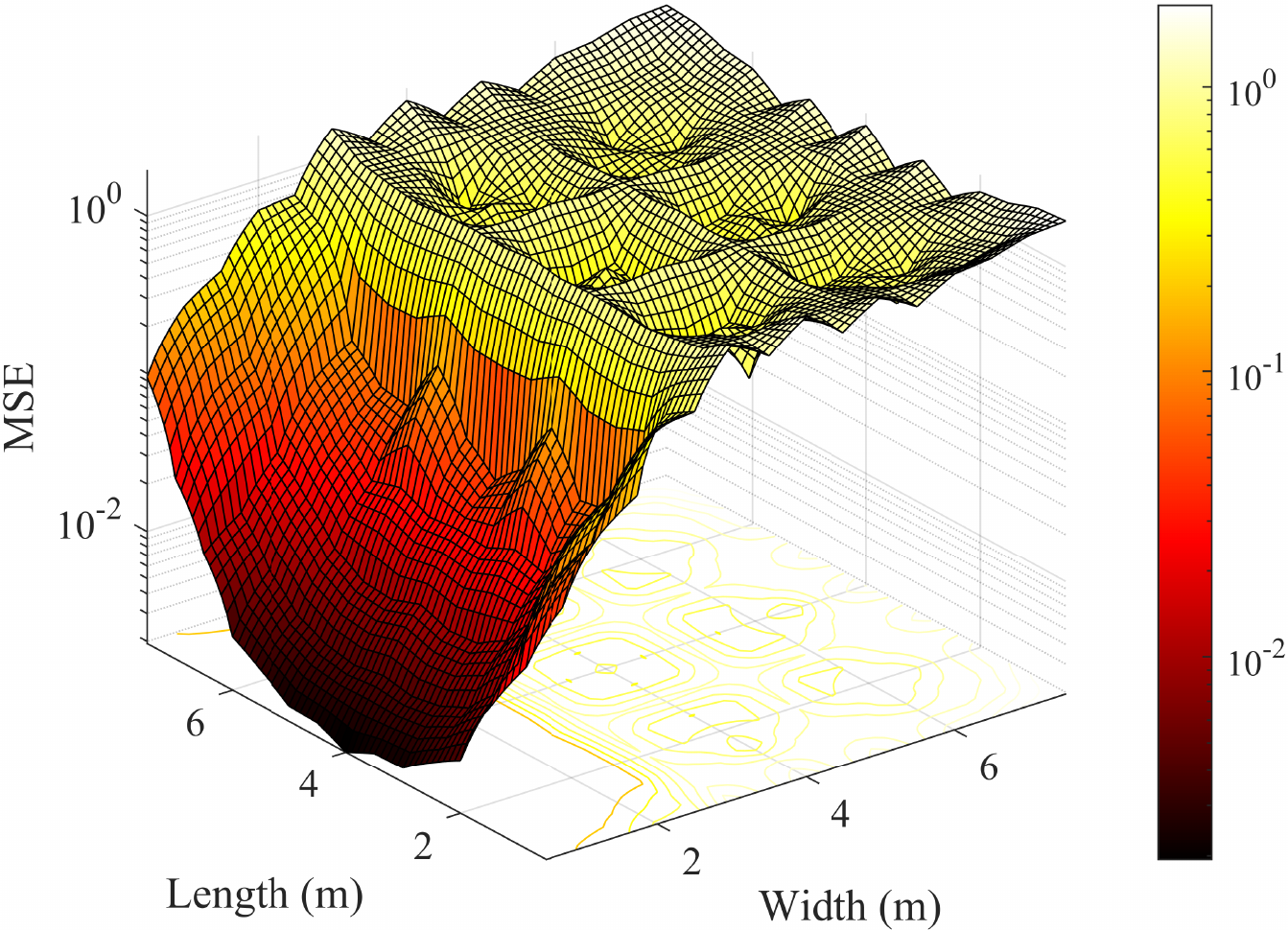}
		\caption{MSE performance at different locations of the room when an IRS is deployed on the wall of width = 0.}
		\label{Fig:Map}
	\end{minipage}
	%	\vspace{-0.2cm}
\end{figure}
%\begin{figure}[t]
%	\centering
%	%	\includegraphics[width=0.45\textwidth]{Iteration.pdf}
%	\includegraphics[width=0.55\textwidth]{Iteration.pdf}
%	\caption{The convergence performance of the proposed AO algorithm.}
%	\label{Fig:Iteration}
%\end{figure}
%\begin{figure}[t]
%	\centering
%	%	\includegraphics[width=0.45\textwidth]{MSEmap.pdf}
%	\includegraphics[width=0.55\textwidth]{MSEmap.pdf}
%	\caption{The MSE performance at different locations of the room when an IRS is deployed on the wall of width = 0.}
%	\label{Fig:Map}
%\end{figure}

% 介绍一次迭代，并绘制三维图
To start with, the convergence performance of the proposed algorithm is illustrated in Fig.~\ref{Fig:Iteration}, where $64$ IRS units facilitate the MIMO VLC system for MSE reduction.
The correlation matrices of the signal vector and noise vector are $\boldsymbol{R}_{\boldsymbol{x}} =\boldsymbol{I}_{N_s}$ and $\boldsymbol{R}_{\boldsymbol{\omega}} =10^{-14} \boldsymbol{I}_{N_r}$, respectively.
As shown in this figure, the MSE of each scheme converges to a positive value after around $40$ iterations. 
Notably, the MSE of the proposed algorithm is smaller than that of the no-IRS scheme, demonstrating the effectiveness of the optical IRS to reduce MSE.
Moreover, considering that the VLC channel gain depends on the locations of the transceivers, we investigate the MSE performance at different locations in Fig.~\ref{Fig:Map}.
It is observed that the MSE at the locations close to the IRS can be less than $10^{-2}$, while the MSE is larger than $10^{-1}$ in the center of the room.
This result shows that the MSE of MIMO VLC can be improved significantly by the optical IRS, while the effect of IRS becomes weaker when the receiver is away from the IRS.

% 介绍SNR的影响
Next, we investigate the performance of the proposed AO algorithm versus different SNR values.
Specifically, two metrics, namely the MSE calculated by~(\ref{Eq:MSE_final}) and the BER obtained by averaging $10^5$ independent experiments, are simulated under the defined $\text{SNR} \triangleq 10^{-13}\sigma_{x}^2/\sigma_{\omega}^2$.
As shown in Fig.~\ref{Fig:MSEvsSNR}, the MSEs of the IRS-aided schemes are lower than the corresponding no-IRS schemes.
Moreover, the proposed algorithm also outperforms other baselines for precoding and signal detection.
The MSE reduction of the proposed AO algorithm is mainly due to two factors, namely the IRS gain maximization and the joint optimization of IRS configuration with precoding and detection.
In addition, the BER performance is investigated in Fig.~\ref{Fig:BERvsSNR}, where the AWGN channel model is adopted and several baselines are given for comparison.
The result shows that the BER of the proposed AO algorithm decreases rapidly with the increase of SNR, and the SNR gaps between it and other IRS-aided schemes are more than $5$ dB, revealing that the proposed algorithm is more effective for reducing the BER.
\begin{figure}[t]
	\centering
	\begin{minipage}[b]{0.45\linewidth}
		\centering
		\includegraphics[width=1\textwidth]{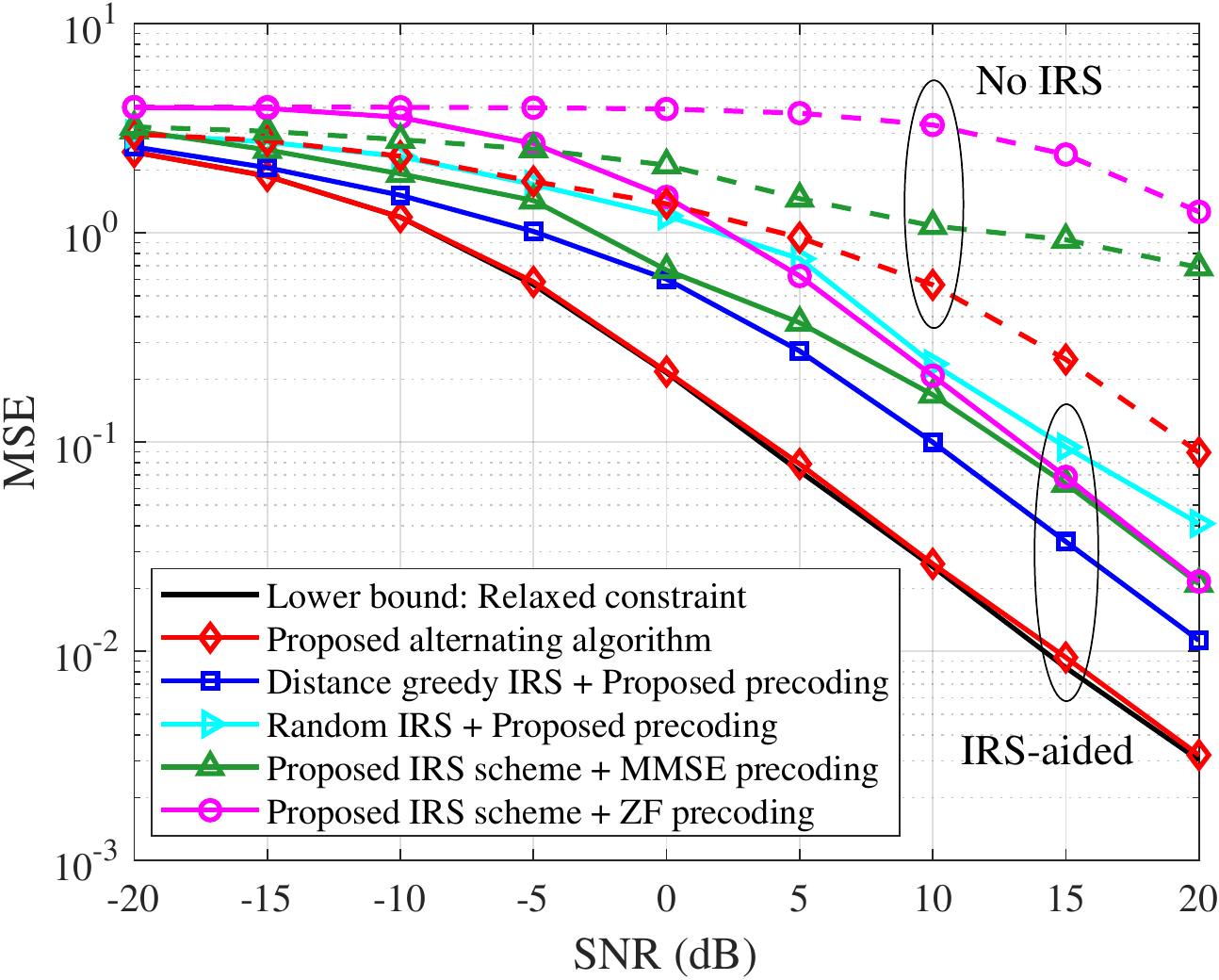}
		\caption{MSE performance of the proposed AO algorithm.}
		\label{Fig:MSEvsSNR}
	\end{minipage}
	\hspace{0.25cm}
	\begin{minipage}[b]{0.45\linewidth}
		\centering
		\includegraphics[width=1\textwidth]{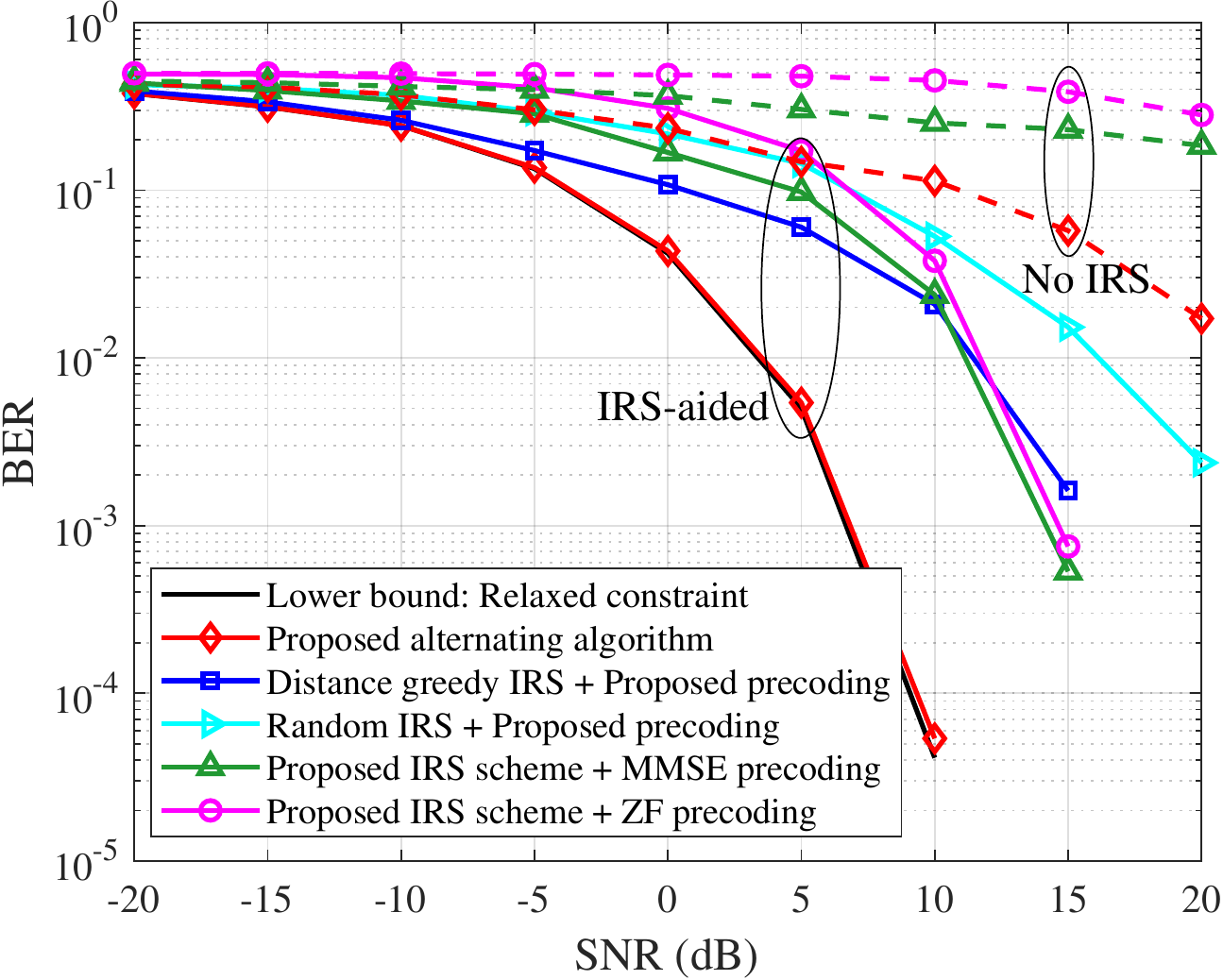}
		\caption{BER performance of the proposed AO algorithm.}
		\label{Fig:BERvsSNR}
	\end{minipage}
	%	\vspace{-0.2cm}
\end{figure}
%\begin{figure}[t]
%	\centering
%%	\includegraphics[width=0.45\textwidth]{MSEvsSNR.pdf}
%	\includegraphics[width=0.55\textwidth]{MSEvsSNR.pdf}
%	\caption{The MSE performance of the proposed AO algorithm.}
%	\label{Fig:MSEvsSNR}
%\end{figure} 
%\begin{figure}[t]
%	\centering
%%	\includegraphics[width=0.45\textwidth]{BERvsSNR.pdf}
%	\includegraphics[width=0.55\textwidth]{BERvsSNR.pdf}
%	\caption{The BER performance of the proposed AO algorithm.}
%	\label{Fig:BERvsSNR}
%\end{figure}

% 介绍N的影响
The number of IRS units $N$ critically determines the performance of IRS~\cite{zhang2020capacity}.
In Fig.~\ref{Fig:Number}, the MSE of several baselines, including ZF/MMSE precoding schemes and random IRS scheme, are simulated under different $N$ that varies from $0$ to $96$.
%Specifically, when the number of IRS units is $48$, the MSE decrements from the assistance of IRS and the coding scheme are nearly $6$ dB and $7.8$ dB, respectively.
As discussed before, the MSE gain of IRS-aided MIMO VLC mainly comes from the IRS power gain and the precoding gain.
The result shows that the precoding gain is insensitive to $N$ while the IRS power gain increases with $N$.
This is reasonable since the MIMO VLC channel $\boldsymbol{H}$ becomes stronger with more IRS units.
\begin{figure}[t]
	\centering
	\begin{minipage}[b]{0.45\linewidth}
		\centering
		\includegraphics[width=1\textwidth]{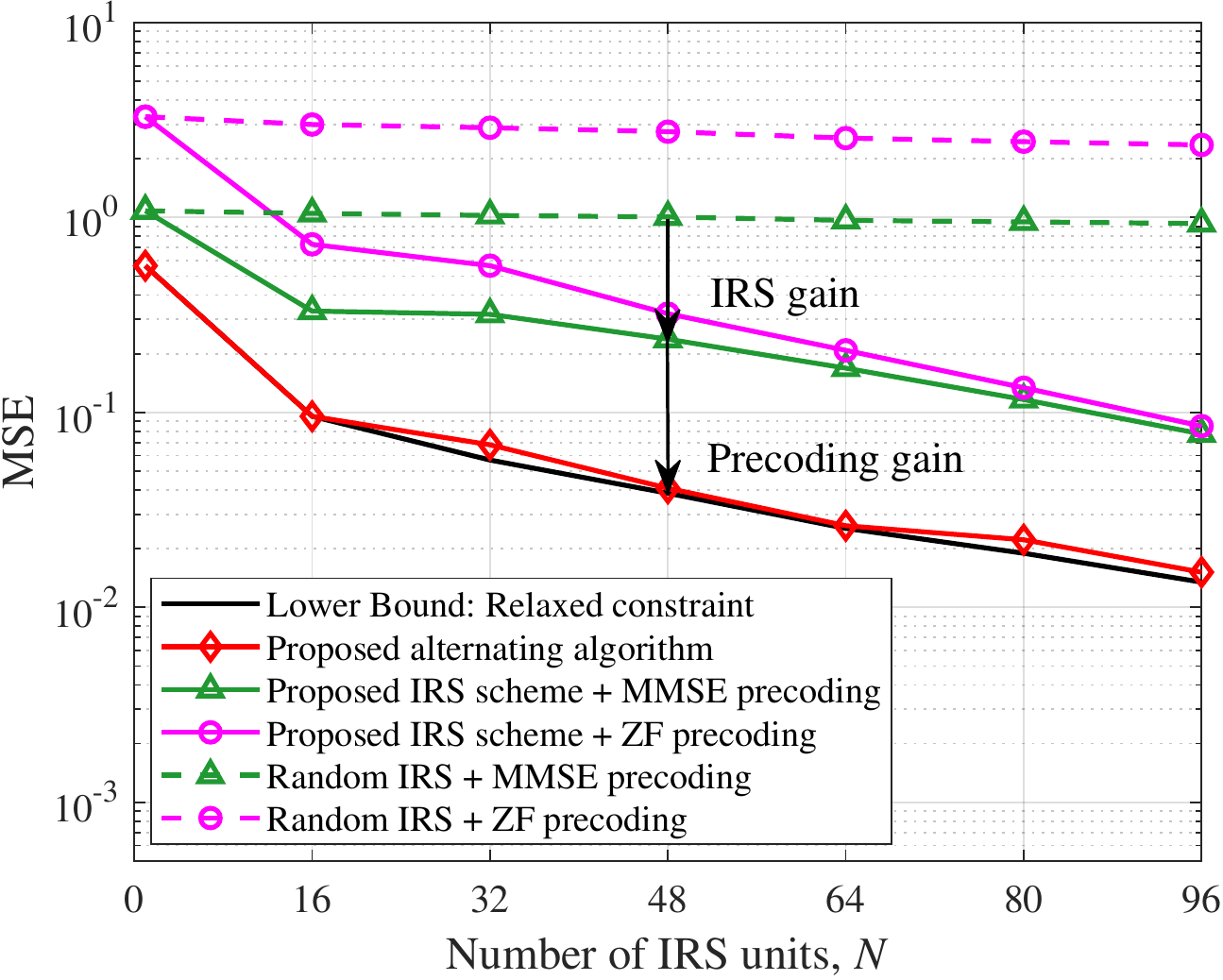}
		\caption{MSE versus the number of IRS units.}
		\label{Fig:Number}
	\end{minipage}
	\hspace{0.25cm}
	\begin{minipage}[b]{0.45\linewidth}
		\centering
		\includegraphics[width=1\textwidth]{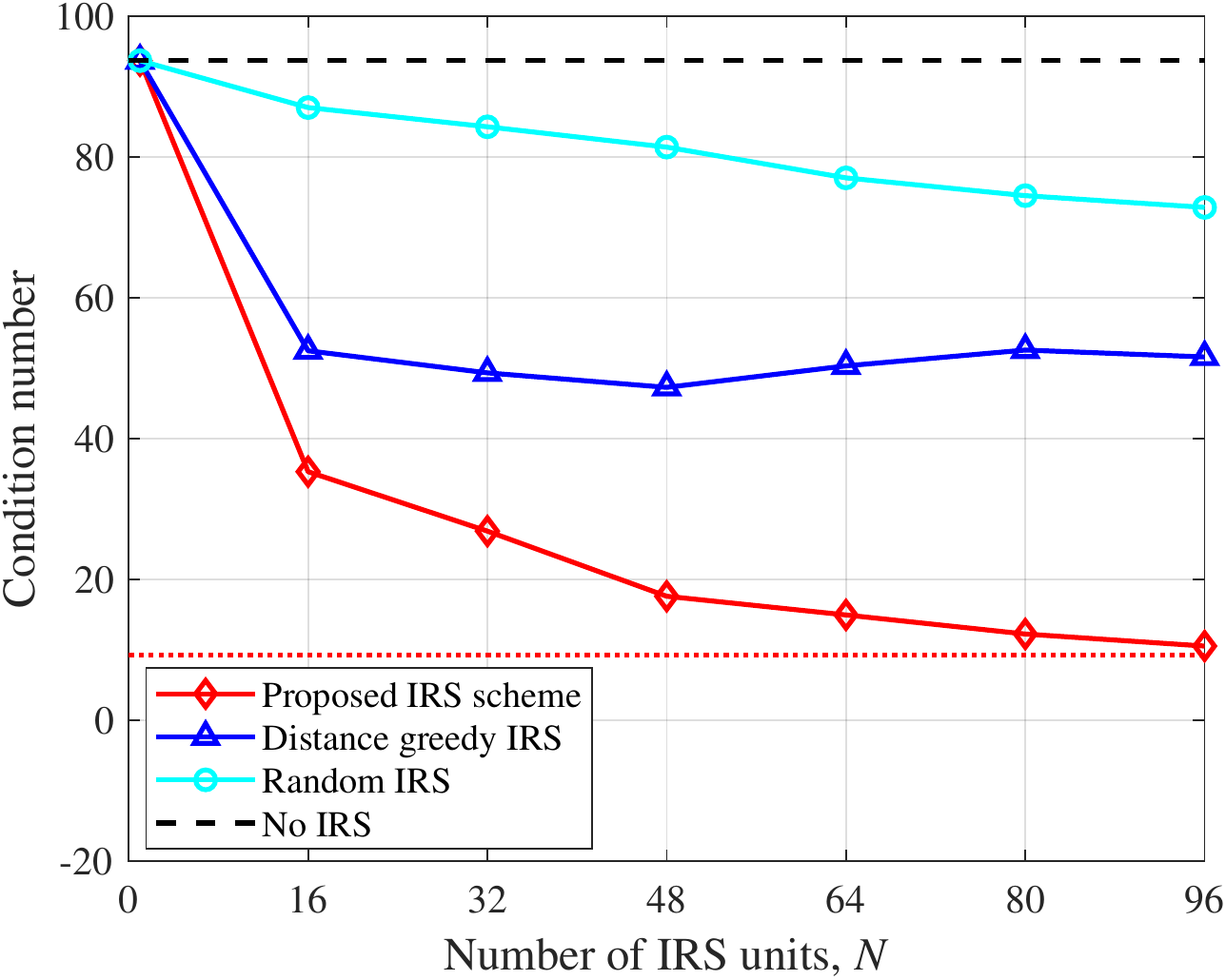}
		\caption{Condition number of MIMO VLC versus the number of IRS units.}
		\label{Fig:CondNumber}
	\end{minipage}
	%	\vspace{-0.2cm}
\end{figure}
%\begin{figure}[t]
%	\centering
%%	\includegraphics[width=0.45\textwidth]{MSEvsNumber.pdf}
%	\includegraphics[width=0.55\textwidth]{MSEvsNumber.pdf}
%	\caption{The MSE performance versus the number of IRS units.}
%	\label{Fig:Number}
%\end{figure}
%\begin{figure}[t]
%	\centering
%%	\includegraphics[width=0.45\textwidth]{CondvsNumber.pdf}
%	\includegraphics[width=0.55\textwidth]{CondvsNumber.pdf}
%	\caption{Condition number of MIMO VLC versus the number of IRS units.}
%	\label{Fig:CondNumber}
%\end{figure}
The performance of MIMO VLC suffers greatly from strong channel correlation since the large MIMO channel matrix condition number results in low multiplexing gain~\cite{ying2015joint}.
Owing to the capability of reconfiguring the wireless environment, IRS has the potential to reduce the condition number of MIMO channels~\cite{zhang2020capacity}.
As depicted in Fig.~\ref{Fig:CondNumber}, the condition numbers of three IRS-aided MIMO channels when $N=96$ are $12$, $52$, and $73$, respectively, significantly improving the condition number $94$ of no-IRS scheme.
Therefore, the IRS-aided MIMO VLC significantly enhances the channel multiplexing gain.
\begin{figure}[t]
	\centering
	\includegraphics[width=0.45\textwidth]{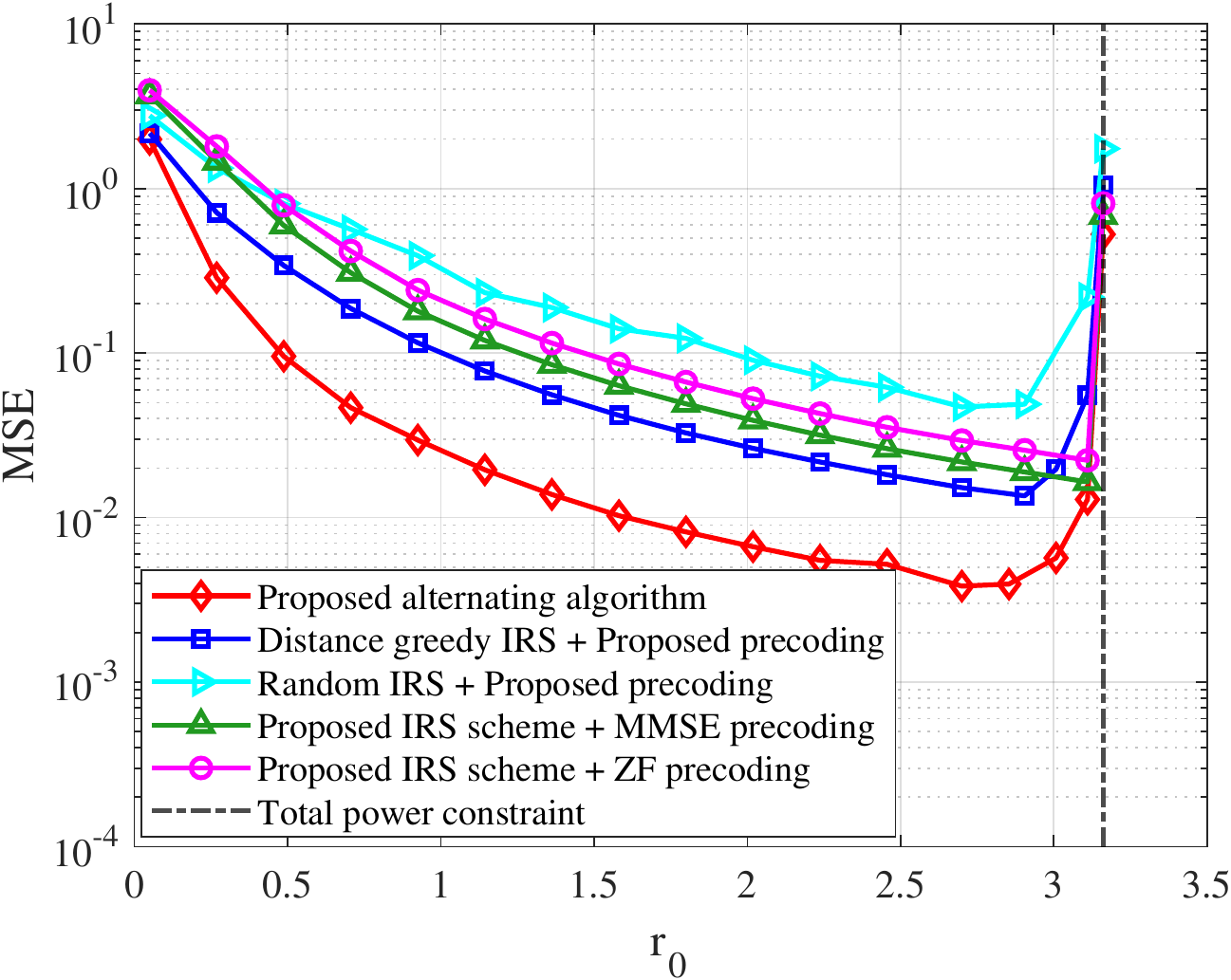}
	\caption{MSE versus the direct current bias $r_0$.}
	\label{Fig:MSEbias}
%	\vspace{-0.2cm}
\end{figure}

% 介绍d的影响
As an evolutionary technology combining both communication and illumination functions, VLC requires an appropriate direct current bias to meet the demand of human eyes.
The effect of the direct bias $r_0$ on MSE is investigated in Fig.~\ref{Fig:MSEbias}, where $r_0$ varies from 0 to $\sqrt{P_{\text{total}}/N_t}$.
The result shows that the MSEs of all schemes decrease versus $r_0$ in the low-bias regime, while it inversely increases when $r_0$ approaches the maximum $\sqrt{P_{\text{total}}/N_t}$.
This is because firstly the feasible set of $\boldsymbol{W}$ is enlarged when $r_0$ increases according to the constraint of~(\ref{Eq:single_power_constraint}), while subsequently the set becomes smaller since the constraint of~(\ref{Eq:total_power_constraint}) becomes active.
Furthermore, the feasible set of $\boldsymbol{W}$ under the condition $r_0 = \sqrt{P_{\text{total}}/N_t}$ is an empty set, which results in the worst MSE performance.
Based on the above, these is an optimal value of the direct current bias $r_0$ that minimizes the MSE.

%****************************************************结论******************************************************************
\vspace{-0.3cm}
\section{Conclusions}
\label{Sec:Conclude}
This paper presented the channel model and signal processing for a new IRS-aided MIMO VLC system, where the IRS configuration, precoding at the transmitter, and detection at the receiver are jointly optimized to minimize the detection MSE.
The non-convex optimization problem was transformed into a more tractable form, and an AO algorithm was proposed to solve it by decoupling the problem into three subproblems, which were shown to be convex and solved efficiently by the PGD algorithm.
In addition, theoretical analysis was provided to ensure the convergence of the proposed algorithm and derive its computational complexity, and furthermore, we discussed simplified algorithms in low/high SNR regimes.
Finally, numerical results showed that the optical IRS can potentially reduce MSE and BER at near-IRS user loations and the proposed AO algorithm outperforms other baselines.
%To sum up, this paper discussed the system modeling and the problem formulation of the IRS-aided MIMO VLC, and except for proposing an theoretically guaranteed optimization algorithm to minimize MSE, various simulations are carried out to verify that IRS can not only improve the MSE and BER performance, but also significantly improve the multiplexing gain of MIMO VLC.

\appendix

\subsection{Global Optimal Point of the Lagrangian Function}
\label{App:dual function}
In pursuing the dual function based on~(\ref{Eq:LagDual_function}), the minimum value of the Lagrangian function $\mathcal{L}( \boldsymbol{V}^*, \boldsymbol{\mu}_1, \boldsymbol{\mu}_2, \boldsymbol{\mu}_3 )$ can be attained by setting $\nabla_{\boldsymbol{V}} \mathcal{L}( \boldsymbol{V}^*, \boldsymbol{\mu}_1, \boldsymbol{\mu}_2, \boldsymbol{\mu}_3 ) = \boldsymbol{0}$,
% Eq.~(\ref{Eq:Lag_0})
%\begin{equation}
%	\label{Eq:Lag_0}
%	\nabla_{\boldsymbol{V}} \mathcal{L}\left( \boldsymbol{V}^*, \boldsymbol{\mu}_1, \boldsymbol{\mu}_2, \boldsymbol{\mu}_3 \right) = \boldsymbol{0},
%\end{equation}
which leads to a globally optimal solution since $\mathcal{L}( \boldsymbol{V}^*, \boldsymbol{\mu}_1, \boldsymbol{\mu}_2, \boldsymbol{\mu}_3 )$ is in a quadratic form of $\text{vec}(\boldsymbol{V})$.
According to the chain rule of matrix differentiation, the derivative is further derived as
\begin{align}
	\label{Eq:dualFunction_derivative}
	\nabla_{\text{vec}\left(\boldsymbol{V}\right)}\ \mathcal{L}\left( \boldsymbol{V}, \boldsymbol{\mu}_1, \boldsymbol{\mu}_2, \boldsymbol{\mu}_3 \right) = \frac{\partial \text{vec}\left( \boldsymbol{H} \right)}{\partial \text{vec}\left( \boldsymbol{V} \right)}^T \!\!\! \frac{\partial \text{MSE}}{\partial \text{vec}\left( \boldsymbol{H} \right)} + \text{vec}\left(\boldsymbol{\mu}_1\boldsymbol{1}_{1\times N_tN_r} \right) - \boldsymbol{\mu}_2 + \boldsymbol{\mu}_3,
\end{align}
where the second term holds due to the equation $\nabla_{\text{vec}(\boldsymbol{V})}(\boldsymbol{a}^T\boldsymbol{V}\boldsymbol{b}) = \text{vec}(\boldsymbol{a}\boldsymbol{b}^T)$.
As for the first term, the MSE metric can be re-expressed as the sum of a quadratic term~(\ref{Eq:MSE_quadratic}) and an affine term, and thereby its derivative with respect to $\text{vec}(\boldsymbol{H})$ can be expressed as
\begin{align}
	\label{Eq:dualFunction_derivative_first}
	\frac{\partial \text{MSE}}{\partial \text{vec}\left( \boldsymbol{H} \right)} = \left( \boldsymbol{U} + \boldsymbol{U}^T \right)\text{vec}\left(\boldsymbol{H}\right) - 2\text{vec}\left(\boldsymbol{Q}^{(t)T}\boldsymbol{R}_{\boldsymbol{x}}\boldsymbol{W}^{(t)T}\right),
\end{align}
where the equation holds due to $\nabla_{\boldsymbol{H}}\text{tr}(\boldsymbol{CH}) = \boldsymbol{C}^T$ and $\boldsymbol{U}$ is introduced as
\begin{align}
	\label{Eq:dualFunction_simplification}
	\boldsymbol{U} = \left( \boldsymbol{W}^{(t)}\boldsymbol{R}_{\boldsymbol{x}}\boldsymbol{W}^{(t)T} \right)\otimes\left( \boldsymbol{Q}^{(t)T}\boldsymbol{Q}^{(t)} \right).
\end{align}
On the other hand, $\boldsymbol{H}$ is related to $\boldsymbol{V}$ through~(\ref{Eq:H_H1_H2}) and~(\ref{Eq:V_linear}), and therefore its derivative with respect to $\text{vec}(\boldsymbol{V})$ results in a Jacobian matrix given by
\begin{equation}
	\label{Eq:dualFunction_derivative_second}
	\frac{\partial \text{vec}\left( \boldsymbol{H} \right)}{\partial \text{vec}\left( \boldsymbol{V} \right)} = \text{diag}\left( \boldsymbol{H}^{\textit{NLoS}} \right)^T.
\end{equation}

By substituting~(\ref{Eq:dualFunction_derivative_first}) and~(\ref{Eq:dualFunction_derivative_second}) into~(\ref{Eq:dualFunction_derivative}) and setting the derivative of the Lagrangian function to zero, the globally optimal solution can be obtained as
\begin{align}
	\label{Eq:dualFunction_derivative_0}
	\text{vec}\left( \boldsymbol{V}^* \right) =& \boldsymbol{Z}^{-1} \Bigg\{ \text{diag}\left( \boldsymbol{H}^{\textit{NLoS}} \right) \Big[2\text{vec}\left(\boldsymbol{Q}^{(t)T}\boldsymbol{R}_{\boldsymbol{x}}\boldsymbol{W}^{(t)T}\right) - \left( \boldsymbol{U} + \boldsymbol{U}^T \right)\text{vec}\left( \boldsymbol{H}_1 \right)\Big] \notag\\
	& - \text{vec}\left(\boldsymbol{\mu}_1\boldsymbol{1}_{1\times N_tN_r} \right) + \boldsymbol{\mu}_2 - \boldsymbol{\mu}_3 \Bigg\},
\end{align}
%\begin{align}
%	\label{Eq:dualFunction_derivative_0}
%	&\text{vec}\left( \boldsymbol{V}^* \right) = \boldsymbol{Z}^{-1} \Bigg\{ 2 \text{diag}\left( \boldsymbol{H}^{\textit{NLoS}} \right)\text{vec}\left(\boldsymbol{Q}^{(t)T}\boldsymbol{R}_{\boldsymbol{x}}\boldsymbol{W}^{(t)T}\right) \notag\\
%	& - \text{vec}\left(\boldsymbol{\mu}_1\boldsymbol{1}_{1\times N_tN_r} \right) - \text{diag}\left( \boldsymbol{H}^{\textit{NLoS}} \right)\left( \boldsymbol{U} + \boldsymbol{U}^T \right)\text{vec}\left( \boldsymbol{H}_1 \right) \notag\\
%	& + \boldsymbol{\mu}_2 - \boldsymbol{\mu}_3 \Bigg\}.
%\end{align}
where $\boldsymbol{Z}$ is a constant matrix given by
\begin{equation}
	\label{Eq:Z}
	\boldsymbol{Z} = \text{diag}\left( \boldsymbol{H}^{\textit{NLoS}} \right) \left( \boldsymbol{U} + \boldsymbol{U}^T \right) \text{diag}\left( \boldsymbol{H}^{\textit{NLoS}} \right)^T.
\end{equation}

\subsection{Gradient Derivation of the Dual Function}
\label{App:gradient}
Given that $g( \boldsymbol{\mu}_1, \boldsymbol{\mu}_2, \boldsymbol{\mu}_3 ) = \mathcal{L}( \boldsymbol{V}^*, \boldsymbol{\mu}_1, \boldsymbol{\mu}_2, \boldsymbol{\mu}_3 )$ and $\boldsymbol{V}^*$ is with respect to $\boldsymbol{\mu}_1$, $\boldsymbol{\mu}_2$, and $\boldsymbol{\mu}_3$ according to~(\ref{Eq:dualFunction_derivative_0}), the gradient of the dual function is related to both the MSE metric and the Lagrangian multipliers.
Therefore, the dual function needs to compute the partial derivative of the vector $\boldsymbol{\mu}_1$, which leads to the following result,
\begin{align}
	\label{Eq:dual_mu1_gradient}
	\nabla_{\boldsymbol{\mu}_1}\ g\left( \boldsymbol{\mu}_1, \boldsymbol{\mu}_2, \boldsymbol{\mu}_3 \right) = \boldsymbol{V}^*\boldsymbol{1}_{N_tN_r\times 1} - \boldsymbol{1}_{N\times 1} + \frac{\partial \text{vec}\left( \boldsymbol{V}^* \right)}{\partial \boldsymbol{\mu}_1}^T \!\! \left[ \frac{\partial \text{MSE}\left( \boldsymbol{V}^* \right)}{\partial \text{vec}\left( \boldsymbol{V}^* \right)}\! +\! \text{vec}\left( \boldsymbol{\mu}_1 \boldsymbol{1}_{1\times N_tN_r} \right) \!-\! \boldsymbol{\mu}_2\! +\! \boldsymbol{\mu}_3 \right]\!,
\end{align}
where the Jacobian matrix is expressed as
\begin{equation}
	\label{Eq:dual_vecV_mu}
	\frac{\partial \text{vec}\left( \boldsymbol{V}^* \right)}{\partial \boldsymbol{\mu}_1} = -\boldsymbol{Z}^{-1}\left(\boldsymbol{1}_{N_tN_r\times 1}\otimes\boldsymbol{I}_{N} \right).
\end{equation}

Similarly, Jacobian matrices for the other two multipliers can be obtained as
\begin{equation}
	\label{Eq:dual_vecV_mu23}
	\frac{\partial \text{vec}\left( \boldsymbol{V}^* \right)}{\partial \boldsymbol{\mu}_2} = \boldsymbol{Z}^{-1},\quad \frac{\partial \text{vec}\left( \boldsymbol{V}^* \right)}{\partial \boldsymbol{\mu}_3} = -\boldsymbol{Z}^{-1},
\end{equation}
which are both symmetric since $(\boldsymbol{Z}^{-1})^T = \boldsymbol{Z}^{-1}$.
Consequently, the partial derivatives with respect to $\boldsymbol{\mu}_2$ and $\boldsymbol{\mu}_3$ are given by
\begin{align}
	\label{Eq:dual_mu2_gradient}
	\nabla_{\boldsymbol{\mu}_2}\ g\left( \boldsymbol{\mu}_1, \boldsymbol{\mu}_2, \boldsymbol{\mu}_3 \right) = -\text{vec}\left( \boldsymbol{V}^* \right) + \boldsymbol{Z}^{-1}\left[ \frac{\partial \text{MSE}\left( \boldsymbol{V}^* \right)}{\partial \text{vec}\left( \boldsymbol{V}^* \right)} + \text{vec}\left( \boldsymbol{\mu}_1 \boldsymbol{1}_{1\times N_tN_r} \right) - \boldsymbol{\mu}_2 + \boldsymbol{\mu}_3 \right],
\end{align}
and
\begin{align}
	\label{Eq:dual_mu3_gradient}
	\nabla_{\boldsymbol{\mu}_3}\ g\left( \boldsymbol{\mu}_1, \boldsymbol{\mu}_2, \boldsymbol{\mu}_3 \right) = \text{vec}\left( \boldsymbol{V}^* \right) - \boldsymbol{1}_{NN_tN_r\times 1} - \boldsymbol{Z}^{-1}\left[ \frac{\partial \text{MSE}\left( \boldsymbol{V}^* \right)}{\partial \text{vec}\left( \boldsymbol{V}^* \right)} + \text{vec}\left( \boldsymbol{\mu}_1 \boldsymbol{1}_{1\times N_tN_r} \right) - \boldsymbol{\mu}_2 + \boldsymbol{\mu}_3 \right],
\end{align}
respectively, where the derivative of MSE with respect to $\text{vec}(\boldsymbol{V}^*)$ is obtained as~(\ref{Eq:dualFunction_derivative_first}) and~(\ref{Eq:dualFunction_derivative_second}).

%\cite{*}
\bibliographystyle{abbrv}
\bibliography{IEEEabrv,reference}

% that's all folks
%%%%%%%%%% If using BibTeX:
%\bibliography{reference}
\end{document}